\documentclass[sigconf]{aamas} 

\pdfoutput=1 

\usepackage{balance} 
\usepackage{amsmath}
\usepackage{mathtools}
\usepackage{thm-restate}
\usepackage{enumerate}
\usepackage[shortlabels,inline]{enumitem}
\usepackage{xcolor}
\usepackage{proof}
\usepackage{array}
\usepackage{xspace}
\usepackage{xparse}
\usepackage{algpseudocode}
\usepackage{tikz,pgfplots}
\usetikzlibrary{shapes}
\usetikzlibrary{automata, arrows.meta, positioning,decorations.markings,shapes.misc,calc, quotes}
\usepackage{xfrac}
\usepackage{caption}
\usepackage[english]{babel}
\usepackage{stmaryrd}
\usepackage{graphicx}
\usepackage{multirow}
\usepackage{bbm}
\usepackage[normalem]{ulem}
\usepackage{sidecap}

\setlist[description]{leftmargin=1cm,labelindent=0.5cm}

\newtheorem{problem}{Problem}
\newtheorem{assumption}{Assumption}

\newcommand{\stam}[1]{}

\newcommand{\AlgApprox}{\textsf{AlgApprox}\xspace}
\newcommand{\AlgExact}{\textsf{AlgExact}\xspace}

\newcommand{\tup}[1]{\ensuremath{\left\langle #1 \right\rangle}}
\newcommand{\set}[1]{\ensuremath{\left\lbrace #1 \right\rbrace}}

\newcommand{\pr}[2]{\ensuremath{\mathbb{P}_{#1}^{#2}}}

\newcommand{\distr}{\ensuremath{\Delta}}

\newcommand{\closure}[1]{\ensuremath{\mathsf{cl}\left[{#1}\right]}}

\newcommand{\dH}{d_h}
\newcommand{\interior}[1]{\ensuremath{\left\langle{#1}\right\rangle}}

\newcommand{\N}{\mathbb{N}}

\newcommand{\zug}[1]{\langle #1 \rangle}

\mathchardef\mhyphen="2D

\newcommand{\V}{\ensuremath{V}}
\newcommand{\VC}{\ensuremath{{\V_\mathsf{c}}}}
\newcommand{\VR}{\ensuremath{{\V_\mathsf{r}}}}
\newcommand{\VZ}{\ensuremath{V_0}}
\newcommand{\VO}{\ensuremath{V_1}}
\newcommand{\E}{\ensuremath{E}}
\newcommand{\EC}{\ensuremath{\E}}
\newcommand{\ER}{\ensuremath{\delta}}
\newcommand{\pathsfin}{\ensuremath{\mathit{Paths_{\mathsf{fin}}}}}
\newcommand{\pathsfinc}{\ensuremath{\mathit{Paths}_{\mathsf{fin}}^{\mathsf{c}}}}
\newcommand{\pathsinf}{\ensuremath{\mathit{Paths_{\mathsf{inf}}}}}
\newcommand{\sched}{\ensuremath{\theta}}
\renewcommand{\succ}{\mathit{Succ}}
\NewDocumentCommand{\dist}{g}{\textsf{dist}\IfNoValueTF{#1}{}{(#1)}}

\newcommand{\PZ}{Player~$0$\xspace}
\newcommand{\PO}{Player~$1$\xspace}
\newcommand{\PI}{Player~$j$\xspace}

\newcommand{\pol}{\ensuremath{\pi}}
\newcommand{\polZ}{\ensuremath{\sigma}}
\newcommand{\polO}{\ensuremath{\tau}}
\newcommand{\winpolZ}{\ensuremath{\Pi_\mathsf{R}}}
\newcommand{\winpolO}{\ensuremath{\Pi_\mathsf{S}}}

\newcommand{\thres}{\ensuremath{\mathit{Th}}}
\newcommand{\RT}{\ensuremath{\mathit{RT}}}
\newcommand{\val}{\ensuremath{\mathit{val}}}
\newcommand{\coval}{\sval}
\newcommand{\rval}{\ensuremath{\mathsf{r}\mhyphen\val}}
\newcommand{\sval}{\ensuremath{\mathsf{s}\mhyphen\val}}

\newcommand{\winpol}[1]{\ensuremath{\Pi_{#1}}}
\newcommand{\grid}{\Delta_\alpha}
\newcommand{\abs}[1]{\llbracket{#1}\rrbracket}

\newcommand{\G}{{\mathcal G}}

\renewcommand{\next}{\ensuremath{\mathbf{X}}}

\newcommand{\spec}{\varphi}

\NewDocumentCommand{\reach}{g}{\ensuremath{\mathit{Reach}\IfNoValueTF{#1}{}{(#1)}}}
\NewDocumentCommand{\avoid}{g}{\ensuremath{\mathit{Avoid}\IfNoValueTF{#1}{}{(#1)}}}
\newcommand{\safe}{\ensuremath{\mathit{Safe}}}

\newcommand{\optcrv}{\rval^*}
\newcommand{\coptcrv}{\sval^*}

\newcommand{\ord}{\prec}
\newcommand{\setdist}{d_h}
\newcommand{\bellman}{\mathcal{T}}

\newcommand{\last}{\operatorname{last}}

\tikzstyle{rond}=[draw,circle,minimum size=5mm, inner sep = 1pt]
\tikzstyle{rect} = [draw, rectangle, minimum size = 5mm, inner sep = 1pt]
\tikzstyle{rectw} = [draw, fill = green!90!black, rectangle, minimum size = 5mm, inner sep = 1pt]
\tikzstyle{rectl} = [draw, fill = red!95!black, rectangle, minimum size = 5mm, inner sep = 1pt]
\tikzstyle{diam} = [draw, diamond, minimum size = 5mm, inner sep = 1pt]

\setcopyright{ifaamas}
\acmConference[AAMAS '25]{Proc.\@ of the 24th International Conference
on Autonomous Agents and Multiagent Systems (AAMAS 2025)}{May 19 -- 23, 2025}
{Detroit, Michigan, USA}{A.~El~Fallah~Seghrouchni, Y.~Vorobeychik, S.~Das, A.~Nowe (eds.)}
\copyrightyear{2025}
\acmYear{2025}
\acmDOI{}
\acmPrice{}
\acmISBN{}

\acmSubmissionID{}

\title{Bidding Games on Markov Decision Processes\\ with Quantitative Reachability Objectives}

\author{Guy Avni}
\affiliation{
  \institution{University of Haifa}
  \city{Haifa}
  \country{Israel}}
\email{gavni@cs.haifa.ac.il}

\author{Martin Kure\v{c}ka}
\affiliation{
  \institution{Masaryk University}
  \city{Brno}
  \country{Czechia}}
\email{kurecka.m@mail.muni.cz}

\author{Kaushik Mallik}
\authornote{Part of the research was done when the author was at the Institute of Science and Technology Austria~(ISTA), Austria.}
\affiliation{
  \institution{IMDEA Software Institute}
  \city{Madrid}
  \country{Spain}}
\email{kaushik.mallik@imdea.org}

\author{Petr Novotn\'{y}}
\affiliation{
  \institution{Masaryk University}
  \city{Brno}
  \country{Czechia}}
\email{petr.novotny@fi.muni.cz}

\author{Suman Sadhukhan}
\affiliation{
  \institution{University of Haifa}
  \city{Haifa}
  \country{Israel}}
\email{ssadhukh@campus.haifa.ac.il}

\begin{abstract}
{\em Graph games} are fundamental in strategic reasoning of multi-agent systems and their environments.
We study a new family of graph games which combine stochastic environmental uncertainties and auction-based interactions among the agents, formalized as bidding games on (finite) \emph{Markov decision processes} (MDP).
Normally, on MDPs, a single decision-maker chooses a sequence of actions, producing a probability distribution over infinite paths.
In bidding games on MDPs, two players---called the \emph{reachability} and \emph{safety} players---bid for the privilege of choosing the next action at each step.
The reachability player's goal is to maximize the probability of reaching a target vertex, whereas the safety player's goal is to minimize it.
These games generalize traditional bidding games on \emph{graphs}, and the existing analysis techniques do not extend.
For instance, the central property of traditional bidding games is the existence of a \emph{threshold} budget, which is a necessary and sufficient budget to guarantee winning for the reachability player.
For MDPs, the threshold becomes a \emph{relation} between the budgets and probabilities of reaching the target.
We devise value-iteration algorithms that approximate thresholds and optimal policies for general MDPs, and compute the exact solutions for acyclic MDPs, 
and show that finding thresholds is at least as hard as solving {\em simple-stochastic games}.
\end{abstract}


\newcommand{\BibTeX}{\rm B\kern-.05em{\sc i\kern-.025em b}\kern-.08em\TeX}

\begin{document}




\maketitle 


\section{Introduction}
{\em Graph games} are fundamental for reasoning about strategic interactions between agents in multi-agent systems, with~\cite{FKL10,KPV16,WG+16} or without~\cite{AHK02} external environments.
Environments, when present, are commonly modeled using stochastic processes, like \emph{Markov decision processes} (MDP)~\cite{SB98} in reinforcement learning (single-agent), and \emph{stochastic games} in the multi-agent setting~\cite{Con92,FV97,CH12}.

We study games which combine stochastic environments with auction-based interactions among players, formalized as \emph{bidding games on MDP arenas}.
An MDP is a graph whose vertices are partitioned into {\em control} vertices and {\em random} vertices, and the game involves the players moving a token along the edges of the graph.
The rules of the game are as follows.
The two players are allocated initial budgets, normalized in a way that their sum is $1$.
When the token reaches a control vertex, an auction is held to determine who chooses where the token goes next.
In these auctions, the players simultaneously submit bids from their available budgets, the higher bidder moves the token and pays his bid amount to the lower bidder.
When the token reaches a random vertex, it automatically moves to one of the successors according to the transition probabilities of the MDP (without affecting the budgets of the players).

We consider the \emph{quantitative reachability} objectives, where the goal of the first player, called the \emph{reachability player}, is to maximize the probability that a given target vertex is reached, and the goal of the second player, called the \emph{safety player}, is to minimize it.

%

\begin{SCfigure}
	\caption{A bidding game on an MDP. 
		The control and random vertices are denoted as circular and diamond-shaped, respectively.
		The probabilistic transitions (marked with arcs) use the uniform distribution.
		The target for the reachability player is $t$.
		The dashed paths are the viable reachability policies for the setting of Ex.~\ref{ex:intro}.
		}
		\label{fig:DAG}
	\begin{tikzpicture}[node distance=0.5cm, every edge quotes/.style = {sloped, font=\small, below}]
		\node[state, diamond, initial above] (a) at (0, 0) {\(a\)};
		\node[state, diamond] (b) [right= of a] {\(b\)};
		\node[state] (c) [below= of a] {\(c\)}; 
		\node[state] (d) [right= of c] {\(d\)}; 
		\node[state, diamond] (e) [right= of d] {\(e\)};
		\node[state] (f) [below= of d] {\(f\)}; 
%
		\node[state] (y) [right= of b] {\(l_2\)}; 
		\node[state, accepting] (w) [below= of e] {\(t\)};
		\node[state] (z) [left = of f] {\(l_1\)};

		\path [->]
		(a) edge (b)
		edge (c)
		(b) edge (d)
		edge (y)
		(c) edge (z)
		edge (d)
		(d) edge (e)
		edge (f)
		(e) edge  (y)
		edge  (w)
		(f) edge (w)
		edge (z)
		(w) edge [loop below] ()
		(y) edge [loop above] ()
		(z) edge [loop below] ();
		\draw[-]		($(a) + (0,-0.6)$) 	arc		(-90:0:.6);
		\draw[-]		($(b) + (0,-0.6)$) 	arc		(-90:0:.6);
		\draw[-]		($(e) + (0,-0.6)$) 	arc		(-90:90:.6);
		
		\path[draw,->,red,thick,densely dashed,opacity=0.6]	($(a)+(0,0.2)$)	--	($(b)+(0.2,0.2)$)	--	($(f)+(0.2,0.2)$)	--	($(w)+(-0.1,0.2)$);
		\path[draw,->,green!50!black,thick,densely dashed,opacity=0.6] ($(a)+(-0.2,-0.1)$)	--	($(c)+(-0.2,-0.2)$)	--	($(e)+(0.2,-0.2)$)	--	($(w)+(0.2,0.1)$);
	\end{tikzpicture}

\end{SCfigure}
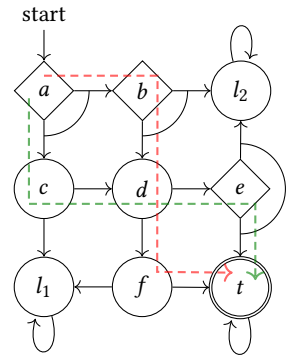
\begin{example}\label{ex:intro}
Consider the game in Fig.~\ref{fig:DAG} and the initial budget allocation $\zug{0.75+\epsilon, 0.25-\epsilon}$ for the two players. We show how the reachability player can reach the target $t$ with probability at least $0.5$. Initially, the token moves randomly from $a$. If it reaches $c$, the reachability player avoids $l_1$ by winning the auction with the bid $0.25$ (which exceeds the opponent's budget) and proceeding to $d$ with new budgets $\zug{0.5+\epsilon, 0.5-\epsilon}$. At $d$, he bids $0.5$, to force the game to $e$, from which $t$ is reached with probability $0.5$. 
If, on the other hand, the token moves to $b$ from $a$, with probability $0.5$ it moves to $l_2$ and the reachability player loses.
If the token reaches $d$, since no biddings were made, the budgets remain $\zug{0.75+\epsilon, 0.25-\epsilon}$.
At $d$, the reachability player bids $0.25$, proceeds to $f$ with budgets $\zug{0.5+\epsilon, 0.5-\epsilon}$, and then bids $0.5$ to force the game to $t$. Thus, each path $b, d, f, t$ and $c, d, e, t$, after $a$, occur with probability $0.5$, and the total probability to reach $t$ from $a$ is $0.5$.
\end{example}

Bidding games on MDPs generalize traditional bidding games on \emph{graphs}, i.e., on MDPs without any random vertices.
Bidding games on graphs have a rich pedigree, going back to the seminal work of Lazarus et al.~\cite{LLPU96,LLPSU99}, followed by a series of extensions to various payment schemes~\cite{AHC19,AHI18,AHZ19,AJZ21}, non-zero-sum games~\cite{MKT18}, discrete bidding aimed for practical applications~\cite{DP10,AAH19,AS22,AM+23}, partial-information games~\cite{AJZ23}, and bidding games with charging~\cite{AGHM24}.

We point out a distinction between traditional and our setups. In the traditional setup, moving the token is memoryless, i.e., there is a winning policy that chooses the same successor from each vertex upon winning the bidding. As seen in the MDP in Ex.~\ref{ex:intro}, the reachability player's choice at $d$ is not memoryless.

\paragraph{Applications}
Our results have an immediate application in {\em auction-based scheduling}~\cite{AMS24}, which is a {\em decentralized} multi-objective decision-making framework. 
In this framework, we are given an arena, modeling the environment, a pair of specifications $\spec_0, \spec_1$, and we want to compute a pair of policies $\sigma_0,\sigma_1$ that will be composed at runtime with the policies bidding against each other at each step for choosing the next action.
The goal is to synthesize $\sigma_0,\sigma_1$ such that their runtime composition fulfills $\spec_0\land\spec_1$.
The advantage is modularity, where the policies can be independently designed, and if one specification changes, only the relevant policy needs to be updated. 
The synthesis algorithms in auction-based scheduling solve two \emph{independent} zero-sum bidding games for $\sigma_0$ and $\sigma_1$ (on the same arena).
As zero-sum bidding games have been studied only on graphs, auction-based scheduling is restricted to graph arenas until now.
Our work will lead to a decentralized solution of multi-objective reachability on MDPs using auction-based scheduling.

\emph{Fair resource allocation} studies how to allocate a collection of items to a set of agents in a {\em fair} manner, where various definitions of fairness exist~\cite{ABFV22,ALMW22}. Bidding games naturally create a fair allocation mechanism~\cite{MKT18,BEF21}, namely allocate an initial budget to each agent, fix an ordering of the items, and hold a bidding for each item one by one. 
Bidding games on MDPs offer \emph{resource-allocation under stochastic uncertainties}~\cite{BEF22,AF+16}: now the agents bid, but the outcome is uncertain (e.g., in online advertisements, the higher bidder gets the ad-slot, but the number of viewers remains stochastic).

\paragraph{Bidding games on graphs}
We briefly survey results on traditional bidding games on graphs~\cite{LLPU96,LLPSU99}. A central quantity in these games are the {\em threshold budgets}: every vertex $v$ is associated with a value $\thres(v) \in [0,1]$ such the reachability player wins from $v$ if his initial budget is strictly larger than $\thres(v)$, and loses (i.e., the safety player wins) if it is strictly smaller than $\thres(v)$.
Furthermore, {\em pure} policies suffice for both players. 
Interestingly, bidding games are equivalent to a class of {\em stochastic games}~\cite{Con92} called {\em random-turn games}~\cite{PSSW07}. In particular, for every bidding game $\G$, the random-turn game $\RT(\G)$ can be obtained where who moves the token at each turn is determined uniformly at random. It is known that the optimal probability with which the reachability player wins from a vertex $v$ equals $1-\thres(v)$. This reduction from bidding games to random turn games implies that computing threshold budgets for bidding games is in NP $\cap$ co-NP. An opposite reduction is still unknown.

\paragraph{Our results: bidding games on MDPs}
We prove thresholds exist for bidding games on MDPs, though their shapes become significantly more complex and reasoning becomes more advanced than bidding games on graphs. This is because thresholds are now \emph{binary relations} between budgets and probabilities, where $(B,p)\in \thres(v)$ if the reachability player can reach the target with every probability $p' < p$ and every budget $B' > B$, and the safety player can avoid the target with every probability $(1-p')>(1-p)$ when the reachability player's budget is $B'<B$. We develop a value-iteration algorithm to find thresholds as described in the next example.

\begin{example}
Fig.~\ref{fig:value iteration illustration} illustrates our value-iteration algorithm. Intuitively, the shaded area in the plot at time $i \in \N$ depicts all the necessary budgets of the reachability players and the respective probabilities of reaching the target $c$ in at most $i$ steps.
For example, for $i=4$ at $a$, if the reachability player has a budget in $(0.5,0.75]$, he can win only one bidding (in $a$) and reach $c$ with probability up to $0.5$ in $2$ steps ($abc$), and if he has a budget in $(0.75,1]$, he can win two biddings (in $a$) and reach $c$ with probability up to $0.75$ in $4$ steps ($ababc$).
Every other path to $c$ is longer than $4$ steps. In the limit, the ``plots'' tend to thresholds (Thm.~\ref{thm:limiting values}), which for $a$ is $\{(B,p)\mid \exists n\geq 1\;.\; B \in (1-2^{-n}, 1-2^{-(n+1)}] \land p\leq 1-2^{-n}\}$.
\end{example}


\begin{figure}
	\begin{tikzpicture}[node distance=0.5cm]
		\node[state,initial]	(a)		at	(0,0)	{$a$};
		\node[state,diamond]	(b)		[below=of a]		{$b$};
		\node[state,accepting]	(c)		[below=of b]		{$c$};
		\node[state]	(d)		[left=of b]		{$d$};
		
		\path[->]
			(a)			edge[bend right]		(b)
						edge[bend right]		(d)
			(b)			edge[bend right]		(a)
						edge					(c)
			(c)			edge[out=0,in=-40,looseness=4]		(c)
			(d)			edge[out=100,in=140,looseness=4]		(d);
		\draw[-]		($(b)+(0,-0.5)$)	arc		(-90:60:0.5);
			
			\begin{scope}[shift={($(c.south)+(-0.5,-0.8)$)},scale=0.7,local bounding box=A]
				\fill[blue!40!white]	(0,0)	rectangle	(1,1);
				\draw[->]	(0,0)	--	(0,1.2);
				\draw[->]	(0,0)	--	(1.2,0);
			\end{scope}
			\node	at	($(A.west)+(-0.7,0)$)		{\footnotesize $i=0,1,\ldots$};
			\node	at	($(A.south east)+(0.15,0)$)	{\footnotesize $B_0$};
			\node	at	($(A.north west)+(0,0.1)$)	{\footnotesize $p$};
			\begin{scope}[shift={($(d.south)+(-0.5,-0.7)$)},scale=0.7,local bounding box=B]
				\draw[blue!40!white,line width=1.5]	(0,0)	--	(1,0)	--	(1,1);
				\draw[->]	(0,0)	--	(0,1.2);
				\draw[->]	(0,0)	--	(1.2,0);
			\end{scope}			
			\node	at	($(B.south)+(-0.1,-0.2)$)		{\footnotesize $i=0,1,\ldots$};
			\node	at	($(B.south east)+(0.15,0)$)	{\footnotesize $B_0$};
			\node	at	($(B.north west)+(0,0.1)$)	{\footnotesize $p$};
			\begin{scope}[shift={(1.5,-0.8)},scale=0.7,local bounding box=E1]
				\draw[blue!40!white,line width=1.5]	(0,0)	--	(1,0)	--	(1,1);
				\draw[->]	(0,0)	--	(0,1.2);
				\draw[->]	(0,0)	--	(1.2,0);
			\end{scope}
			\node	at	($(E1.east)+(-1.3,0)$)		{\footnotesize $i=0$};
			\node	at	($(E1.south east)+(0.15,0)$)	{\footnotesize $B_0$};
			\node	at	($(E1.north west)+(0,0.1)$)	{\footnotesize $p$};
			\node	at	($(E1.north)+(0,0.3)$)		{a};
			\begin{scope}[shift={($(E1.south east)+(0.5,0)$)},local bounding box=F1,scale=0.7]
				\draw[blue!40!white,line width=1.5]	(0,0)	--	(1,0)	--	(1,1);
				\draw[->]	(0,0)	--	(0,1.2);
				\draw[->]	(0,0)	--	(1.2,0);
			\end{scope}
			\node	at	($(F1.north)+(0,0.3)$)		{b};
			\begin{scope}[shift={($(E1.south west)+(0,-1)$)},local bounding box=E2,scale=0.7]
				\draw[blue!40!white,line width=1.5]	(0,0)	--	(1,0)	--	(1,1);
				\draw[->]	(0,0)	--	(0,1.2);
				\draw[->]	(0,0)	--	(1.2,0);
			\end{scope}
			\node	at	($(E2.east)+(-1.3,0)$)		{\footnotesize $i=1$};
			\begin{scope}[shift={($(F1.south west)+(0,-1)$)},local bounding box=F2,scale=0.7]
				\draw[blue!40!white,line width=1.5]	(0,0)	--	(1,0)	--	(1,1);
				\fill[blue!40!white]	(0,0)	rectangle	(1,0.5);
				\draw[->]	(0,0)	--	(0,1.2);
				\draw[->]	(0,0)	--	(1.2,0);
			\end{scope}
			\begin{scope}[shift={($(E2.south west)+(0,-1)$)},local bounding box=E3,scale=0.7]
				\draw[blue!40!white,line width=1.5]	(0,0)	--	(1,0)	--	(1,1);
				\fill[blue!40!white]	(0.5,0)	rectangle	(1,0.5);
				\draw[->]	(0,0)	--	(0,1.2);
				\draw[->]	(0,0)	--	(1.2,0);
			\end{scope}
			\node	at	($(E3.east)+(-1.3,0)$)		{\footnotesize $i=2$};
			\begin{scope}[shift={($(F2.south west)+(0,-1)$)},local bounding box=F3,scale=0.7]
				\draw[blue!40!white,line width=1.5]	(0,0)	--	(1,0)	--	(1,1);
				\fill[blue!40!white]	(0,0)	rectangle	(1,0.5);
				\draw[->]	(0,0)	--	(0,1.2);
				\draw[->]	(0,0)	--	(1.2,0);
			\end{scope}
			\begin{scope}[shift={($(E3.south west)+(0,-1)$)},local bounding box=E4,scale=0.7]
				\draw[blue!40!white,line width=1.5]	(0,0)	--	(1,0)	--	(1,1);
				\fill[blue!40!white]	(0.5,0)	rectangle	(1,0.5);
				\draw[->]	(0,0)	--	(0,1.2);
				\draw[->]	(0,0)	--	(1.2,0);
			\end{scope}
			\node	at	($(E4.east)+(-1.3,0)$)		{\footnotesize $i=3$};
			\begin{scope}[shift={($(F3.south west)+(0,-1)$)},local bounding box=F4,scale=0.7]
				\draw[blue!40!white,line width=1.5]	(0,0)	--	(1,0)	--	(1,1);
				\fill[blue!40!white]	(0,0)	--	
														(0,0.5)	--	(0.5,0.5)	--	
														(0.5,0.75)	-- (1,0.75) --
														(1,0)	--	cycle;
				\draw[->]	(0,0)	--	(0,1.2);
				\draw[->]	(0,0)	--	(1.2,0);
			\end{scope}
			\begin{scope}[shift={($(F1.south east)+(0.7,0)$)},scale=0.7,local bounding box=G1]
				\draw[blue!40!white,line width=1.5]		(0,0)	-- 	(1,0)	--	(1,1);
				\fill[blue!40!white]			(0.5,0)	--
											(0.5,0.5)	--	(0.75,0.5)	--
											(0.75,0.75)	--	(1,0.75)	--
											(1,0)	--	cycle;
				\draw[->]	(0,0)	--	(0,1.2);
				\draw[->]	(0,0)	--	(1.2,0);
			\end{scope}
			\node	at	($(G1.north)+(0,0.3)$)		{a};
			\node	at	($(G1.east)+(-1.2,0)$)		{\footnotesize $i=4$};
			\begin{scope}[shift={($(G1.south east)+(0.5,0)$)},local bounding box=H1,scale=0.7]
				\draw[blue!40!white,line width=1.5]	(0,0)	--	(1,0)	--	(1,1);
				\fill[blue!40!white]	(0,0)	--	
														(0,0.5)	--	(0.5,0.5)	--	
														(0.5,0.75)	-- (1,0.75) --
														(1,0)	--	cycle;
				\draw[->]	(0,0)	--	(0,1.2);
				\draw[->]	(0,0)	--	(1.2,0);
			\end{scope}
			\node	at	($(H1.north)+(0,0.3)$)		{b};
			\begin{scope}[shift={($(G1.south west)+(0,-1)$)},local bounding box=G2,scale=0.7]
				\draw[blue!40!white,line width=1.5]	(0,0)	--	(1,0)	--	(1,1);
				\fill[blue!40!white]			(0.5,0)	--
											(0.5,0.5)	--	(0.75,0.5)	--
											(0.75,0.75)	--	(1,0.75)	--
											(1,0)	--	cycle;
				\draw[->]	(0,0)	--	(0,1.2);
				\draw[->]	(0,0)	--	(1.2,0);
			\end{scope}
			\node	at	($(G2.east)+(-1.2,0)$)		{\footnotesize $i=5$};
			\begin{scope}[shift={($(H1.south west)+(0,-1)$)},local bounding box=H2,scale=0.7]
				\draw[blue!40!white,line width=1.5]	(0,0)	--	(1,0)	--	(1,1);
				\fill[blue!40!white]	(0,0)	--	
														(0,0.5)	--	(0.5,0.5)	--	
														(0.5,0.75)	-- (0.75,0.75) --
														(0.75,0.875)	--	(1,0.875)	--
														(1,0)	--	cycle;
				\draw[->]	(0,0)	--	(0,1.2);
				\draw[->]	(0,0)	--	(1.2,0);
			\end{scope}
			\node	at	($(G2.south east)+(0.2,-0.2)$)	{$\vdots$};
			\begin{scope}[shift={($(F4.south east)+(0.7,0.5)$)},local bounding box=G4,scale=0.7]
				\draw[blue!40!white,line width=1.5]	(0,0)	--	(1,0)	--	(1,1);
				\fill[blue!40!white]			(0.5,0)	--
											(0.5,0.5)	--	(0.75,0.5)	--
											(0.75,0.75)	--	(0.875,0.75)	--
											(0.875,0.875)	--	(0.9375,0.875)	--
											(0.9375,0.9375)	--	(0.96875,0.9375)	--
											(0.96875,0.96875)	--	(0.984375,0.96875)	--
											(0.984375,0.984375)	--	(0.9921875,0.984375) --
											(0.9921875,0.9921875)	--	(1,0.9921875)	--
											(1,0)	--	cycle;
				\draw[->]	(0,0)	--	(0,1.2);
				\draw[->]	(0,0)	--	(1.2,0);
			\end{scope}
			\node	at	($(G4.east)+(-1.2,0)$)		{\footnotesize $i=16$};
			\begin{scope}[shift={($(G4.south east)+(0.5,0)$)},local bounding box=H4,scale=0.7]
				\draw[blue!40!white,line width=1.5]	(0,0)	--	(1,0)	--	(1,1);
				\fill[blue!40!white]	(0,0)	--	
									(0,0.5)	--	(0.5,0.5)	--	
									(0.5,0.75)	-- (0.75,0.75)	--
									(0.75,0.875)	--	(0.875,0.875)	--
									(0.875,0.9375)	--	(0.9375,0.9375)	--
									(0.9375,0.96875)	--	(0.96875,0.96875)	--
									(0.96875,0.984375)	--	(0.984375,0.984375)	--
									(0.984375,0.9921875)	--	(0.9921875,0.9921875)	--
									(0.9921875,0.99609375)	--	(1,0.99609375)	--
									(1,0)	--	cycle;
				\draw[->]	(0,0)	--	(0,1.2);
				\draw[->]	(0,0)	--	(1.2,0);
			\end{scope}
			\node	at	($(G4.south east)+(0.2,-0.2)$)	{$\vdots$};
	\end{tikzpicture}
	\caption{Value iteration for the game on the \emph{left} with the sequence of reachability values ($c$ is the target) on the \emph{right}. }
	\label{fig:value iteration illustration}
\end{figure}
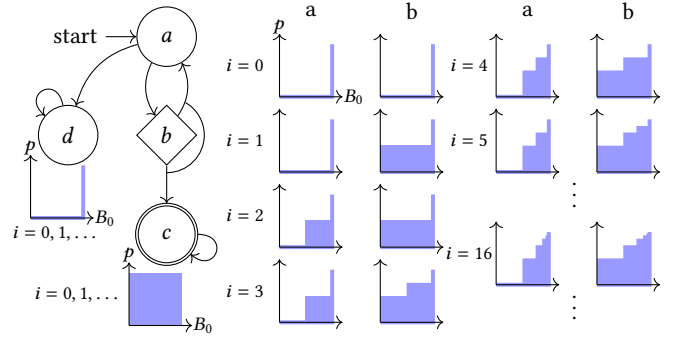  

We summarize our results below.
We consider the problem of deciding if the reachability player can reach the target in a given MDP with a given probability $p$ using a given budget $B$. \textbf{(I)}~For \emph{general MDPs}, the reachability problem remains open. Under the assumption that $(B,p)$ is not exactly on the threshold, we show that the problem is decidable. The time and space complexity of our algorithm depends on the distance $\epsilon$ of $(B,p)$ from the threshold (infinity norm), the minimum probability $\ER_{\min}$, and the number of vertices $|\V|$, and is given as $\mathcal{O}\left(\frac{|\V|^2}{\epsilon}\log\left({1/\epsilon}\right)^3\ER_{\min}^{-4|\V|}\right)$. Our decision procedure uses an approximated value iteration algorithm to limit the computational complexity. \textbf{(II)}~For \emph{acyclic MDPs}, the reachability problem is decidable in $\mathcal{O}(|\V|^{|\V|})$ time and space, for tree-shaped MDPs, it is decidable in NP $\cap$ co-NP. 

The above assumption that the point \( (B,p) \) does not lie exactly on the threshold is natural in the context of multi-objective decision-making. E.g. for ``classical'' (non-bidding) multi-objective stochastic games, algorithms for determining the winning player assume that the target payoff vector does not lie on the boundary of the Pareto set of achievable payoffs~\cite{MOgames:approx}. 

\medskip

\stam{
\noindent\textbf{Further Related Work}

\paragraph{Multi-objective decision making.} As mentioned earlier, one of our motivations is solving multi-objective problems on MDPs via the auction-based scheduling approach. The alternative approach is to directly synthesize a single policy achieving an acceptable tradeoff between the individual objectives. Such approaches were investigated in the context of multi-objective MDPs~\cite{CMH06,MOMDP,convex-hull-mcts}, stochastic games~\cite{stochgames:vi,MOgames:approx}, and reinforcement learning~\cite{BarettNarayanan:CHVI,JMLR:v15:vanmoffaert14a,MORLguide,deepMORL}.
}

\noindent\textbf{Further Related Work.}
As mentioned earlier, one of our motivations is solving multi-objective problems on MDPs via the auction-based scheduling approach. 
The alternative approach is to directly synthesize a single policy achieving an acceptable tradeoff between the individual objectives as was studied for  
MDPs~\cite{CMH06,MOMDP,convex-hull-mcts}, stochastic games~\cite{stochgames:vi,MOgames:approx}, and reinforcement learning~\cite{BarettNarayanan:CHVI,JMLR:v15:vanmoffaert14a,MORLguide,deepMORL}.
We are the first to study quantitative reachability objectives in bidding games on MDPs. For sure winning, simple reductions to bidding games on graphs where shown~\cite{AHIN19}.


\section{Preliminaries of Markov Decision Processes (MDP)} 
	\noindent\textbf{Syntax.} An MDP is a tuple $\tup{\V,\VC,\VR,\EC,\ER}$,
	where $\V$ is a finite set of vertices,
	$\VC$ and $\VR$ are the \emph{control} and \emph{random} vertices such that $\VC\cup \VR=\V$ and $\VC\cap \VR=\emptyset$,
	$\EC \colon \VC\to 2^\VR$ is the \emph{control transition function}, and
	$\ER \colon \VR\to \distr(\VC)$ is the \emph{random transition function}, where $\distr(\VC)$ is the set of all probability distributions over $\VC$.
	The set of \emph{successors} of vertex $v$ will be denoted as $\succ(v)$, where $\succ(v)\coloneqq \EC(v)$ if $v\in\VC$ and $\succ(v)\coloneqq \set{v'\in\V\mid \ER(v)(v')>0}$ if $v\in\VR$.
	A vertex $v$ is called \emph{sink} if $\succ(v)=\set{v}$.
	
	\medskip
	\noindent\textbf{Convention for figures.} MDPs are depicted as transition diagrams with circular nodes representing control vertices and diamond-shaped nodes representing random vertices.
	The target vertices are depicted in double circles.
	Random transitions with \emph{uniform} probability distributions are marked using an arc between them.

	\medskip
	\noindent\textbf{Semantics.}
	Semantics of MDPs are summarized below; details can be found in standard textbooks~\cite{puterman1990markov}. 
	A \emph{path} of an MDP starting at a given vertex $v\in \V$ is a sequence $v^0v^1\ldots$ with $v^0=v$ and every $v^{i>0}$ being a successor of $v^{i-1}$.
	Paths can be either finite or infinite.
	We write $\pathsfin(M)$ and $\pathsinf(M)$ to denote, respectively, the set of all finite and infinite paths, and write $\pathsfinc(M)$ to denote the set of all finite paths that end in a control vertex.
	A \emph{scheduler} is a function $\sched\colon \pathsfinc(M)\to \V$ mapping every finite path $\rho=v^0\ldots v^k$ ending at the control vertex $v^k\in \VC$ to one of its successors; i.e., $\sched(\rho)\in \EC(v^k)$.
	Every scheduler $\sched$ induces a unique probability distribution $\pr{v}{M,\sched}(\cdot)$ over the paths of $M$ with initial vertex $v$.
	
	\medskip
	\noindent\textbf{Specifications.}
	A \emph{specification} $\spec$ over an MDP $M$ is a set of infinite paths of $M$.
	We will consider reachability and safety specifications of both bounded and unbounded variants, defined below. 
	Given a set of vertices $T\subseteq \V$ called the \emph{target} vertices, and an integer $h>0$, the \emph{bounded-horizon reachability} specification is the set of paths that visit $T$ \emph{in at most $h$ steps}, i.e., $\reach^{M,h}(T)\coloneqq\set{v^0v^1\ldots\in \pathsinf(M)\mid \exists i\leq h\;.\;v^i\in T}$. The (unbounded) \emph{reachability} specification is the set of paths that \emph{eventually} visit $T$; i.e., $\reach^M(T)\coloneqq\bigcup_h \reach^{M,h}(T)$.
	Dually, given a set of vertices $U\subseteq \V$ called the \emph{unsafe} vertices, and a number $h>0$, the \emph{bounded-horizon safety} specification is the set of paths that avoid $U$ \emph{for at least $h$ steps}, i.e., $\safe^{M,h}(U)\coloneqq\set{v^0v^1\ldots\in \pathsinf(M)\mid \forall i\leq h\;.\;v^i\notin U}$. The (unbounded) \emph{safety} specification is the set of paths that \emph{always} avoid $U$; i.e., $\safe^M(U)\coloneqq\bigcap_h \safe^{M,h}(U)$.
	Reachability and safety specifications---with bounded and unbounded horizons---are complementary to each other, i.e.,  for every $h>0$, $\reach^{M,h}(T) = \V^\omega\setminus\safe^{M,h}(T)$, and $\reach^M(T) = \V^\omega\setminus\safe^M(T)$.

\section{Bidding Games on MDP-s}
	On a given $M$, we consider a zero-sum ``token game'' between two players, who will be referred to as the \emph{reachability} and \emph{safety} players.
	Initially, the token is placed in a given \emph{initial} vertex, and the players are allocated budgets (positive real numbers) whose sum is $1$.
	As convention, we will only specify the reachability player's budget as $B$, and the safety player's budget will be implicit (i.e., $1-B$).
	
	The game is played as follows.
	When the token is in a control vertex $v$, the players independently and simultaneously propose their bids which may not exceed their current budgets.
	Whoever bids higher pays his bid amount to the other player and moves the token to one of $v$'s successors.
	The budgets are updated accordingly:
	assuming the reachability player's budget at $v$ was $B$, if he wins by bidding $b_R$, his new budget will reduce to $B-b_R$, and if the safety player wins by bidding $b_S$, the reachability player's new budget will increase to $B+b_S$.\footnote{This bidding mechanism is known as Richman bidding in the literature. Other bidding mechanisms also exist, but they are left as part of future works.}
	On the other hand, when the token is in a random vertex $v$, it moves to a successor $w$ with probability $\ER(v)(w)$, and the budgets of the players remain unaffected.
	
	The game continues in this fashion forever, generating an infinite path traversed by the token.
	Given a set of target vertices $T$, the goal of the reachability player is to maximize the probability that the token eventually reaches $T$ from the initial vertex, while the goal of the safety player is to minimize this probability.

	\medskip
	\noindent\textbf{Policies and paths.}
	We formalize bidding games on MDPs as follows.
	A  policy of a player is a function of the form $[0,1]\times \pathsfinc(M)\to [0,1]\times \V$, mapping every pair of  available budget $B$ and finite path $v^0\ldots v^k$ to a pair of a bid value $b\leq B$ and a successor of $v^k$.
	We will write $\polZ$ and $\polO$ to represent the policy of the reachability and the safety player, respectively.
	
	Suppose we are given an initial vertex $v$ and an initial budget $B$ of the reachability player (recall that the safety player's initial budget will be $1-B$).
	We will call the pair $\tup{v,B}$ the \emph{initial configuration}.
	Every pair of policies $(\polZ,\polO)$ and the initial configuration $\tup{v,B}$ induce a scheduler $\sched(\polZ,\polO,B)$ as follows:
	if the current path is $\rho\in \pathsfinc(M)$ and the current budget of the reachability player is $B'$, then, denoting $\polZ(B',\rho)=(b_R,u)$ and $\polO(1-B',\rho)=(b_S,w)$, we define the scheduler $\sched$ as follows:
	\begin{itemize}[leftmargin=*]
		\item if $b_R \geq b_S$,\footnote{We assume, arbitrarily, that ties go in favor of the reachability player. In our proofs, we show that it does not matter how ties are resolved.} i.e., if the reachability player wins the bidding, then $\sched(\polZ,\polO,B)(\rho)=u$, and the reachability player's new budget is $B'-b_R$, and
		\item if $b_R<b_S$, i.e., if the safety player wins the bidding, then $\sched(\polZ,\polO,B)(\rho)=w$, and the reachability player's new budget is $B'+b_S$.
	\end{itemize}
	We will write $\pr{v,B}{\polZ,\polO}$ instead of $\pr{v}{\sched(\polZ,\polO,B)}$ to denote the probability distribution over the set of infinite paths starting at vertex $v$.
	
	\medskip
	\noindent\textbf{Winning conditions.}
	Let $\tup{v,B}$ be an initial configuration, $\spec$ be a reachability specification (bounded or unbounded),  and $p\in [0,1]$ be the \emph{required probability} for the reachability player to satisfy $\spec$.
	A \emph{winning policy of the reachability player} is a policy $\polZ$ such that for every policy $\polO$ of the safety player, it holds that $\pr{v, B}{\polZ,\polO}(\spec)\geq p$.
	Dually, a \emph{winning policy of the safety player} is a policy $\polO$ such that for every policy $\polZ$ of the reachability player, it holds that $\pr{v, B}{\polZ,\polO}(\spec)\leq p$.
	We will write $\winpolZ(B,p,v,\spec)$ and $\winpolO(B,p,v,\spec)$ to denote the sets of winning policies for reachability and safety players, respectively.
	If $\spec$ is clear, we will simply write $\winpolZ(B,p,v)$ and $\winpolO(B,p,v)$.

	\medskip
	\noindent\textbf{Thresholds.}
	In traditional reachability bidding games on graphs, where probabilities are unnecessary, the threshold of a vertex $v$ is the budget $B$ such that the reachability player wins from $v$ with every budget $B'>B$, and loses with every budget $B'<B$.
	In bidding games on MDPs, thresholds generalize to relations over budgets and probabilities:
	The threshold of $v$ is the set of all pairs $(B,p)$ such that the reachability player wins with every budget greater than $B$  and required probability less than $p$, and loses with every budget less than $B$ and required probability larger than $p$.
	
	\begin{definition}[Threshold]
	For a given vertex $v$, the \emph{threshold} of $v$, written $\thres_v$,
	is the set of all pairs $(B,p)$ such that $\winpolZ(B',p', v)$ is nonempty whenever $B'>B$ and $p'< p$,
	and $\winpolO(B',p', v)$ is nonempty whenever $B'<B$ and $p'>p$.
	\end{definition}
	
	A central question in traditional bidding games is whether thresholds exist, because then it can be \emph{determined} which of the players will win based on the budget allocation, as long as the budget is not exactly equal to the threshold.
	In our case, the existence question of thresholds generalizes to the question of \emph{whether the threshold completely separates}
	the winning points of the two players.

	\begin{definition}[Completely separating thresholds]\label{def:separating threshold}
		The threshold of $v$ is \emph{completely separating} if for
		every point $(B,p)\notin \thres_v$,
		\begin{itemize}[leftmargin=*]
			\item there exists $(B', p') \in \thres_v$
		such that either $B < B'$ and $p > p'$, or $B > B'$ and $p < p'$, and
			\item exactly one of the sets $\winpolZ(B,p,v)$ and $\winpolO(B,p,v)$ is nonempty.
		\end{itemize}
	\end{definition}

	\medskip
	\noindent\textbf{The algorithmic question.}
	We define \emph{problem instances} as tuples of the form $\tup{M,v,T,B,p}$, where $M$ is an MDP, $\tup{v,B}$ is the initial configuration, $T$ is the target, and $p$ is the required probability of satisfying the reachability specification $\reach^M(T)$.
	The subject of this paper is how to decide who wins in a given problem instance.
	\begin{problem}[Quantitative reachability]\label{prob:exact}
		Let $\tup{M,v,T,B,p}$ be a problem instance.
		For a given $j\in \set{\mathsf{R},\mathsf{S}}$, decide if the set $\winpol{j}(B,p,v,\reach^M(T))$ is nonempty.
	\end{problem}
	If $\winpolZ(B,p,v,\reach^M(T))\neq\emptyset$, our decision procedure will produce the witness winning policy for the reachability player as a byproduct; construction of winning policies for the safety player is solved for acyclic MDPs, and remains open for general MDPs.
	We will assume that $T$ is a set of sinks, which is without loss of any generality since the game ends as soon as $T$ is reached.

%


\section{Bounded-Horizon Value Iteration}

We start with the bounded-horizon variant of Prob.~\ref{prob:exact} with horizon $h$.
In this setting, we propose a 2-dimensional
value iteration algorithm for deciding who wins the game.

For reachability, for each vertex $v$, our algorithm computes a monotonically increasing (with respect to ``$\subseteq$'') sequence of ``values'' $\rval_v^0,\ldots,\rval_v^h\subseteq [0,1]^2$, where $\rval_v^i$ will be shown to represent the set of all $(\overline{B},p)$ such that for every $B>\overline{B}$, the reachability player can reach $T$ from the initial configuration $\tup{v,B}$ with probability at least $p$ in at most $i$ steps.
Dually, for safety, for each vertex $v$, our algorithm computes a monotonically decreasing sequence of ``values'' $\sval_v^0,\ldots,\sval_v^h\subseteq [0,1]^2$, where $\sval_v^i$ will be shown to represent the set of all $(\overline{B},p)$ such that for every $B<\overline{B}$, the safety player can avoid $T$ from the initial configuration $\tup{v,B}$ with probability at least $1-p$ for at least $i$ steps.

Clearly, if $v$ is in $T$, the target $T$ will be ``reached'' in zero steps, no matter what $(B,p)$ is, and therefore every $(\overline{B},p)$ belongs to $\rval_v^0$.
In contrast, if $v$ is \emph{not} in $T$, the target $T$ will be reached in zero steps only with probability $p=0$. The points with $\overline{B}=1$ are trivially included to $\rval_v^0$ as well, because $\{B \mid B>\overline{B} = 1\} = \emptyset$. 
By duality, the definition of $\sval_v^0$ is exactly the opposite.

We now consider the case of $i>0$ and $v\notin T$. 
We take the perspective of the reachability player; the case of safety is similar.
Consider the following two cases.
(a)~Suppose $v\in \VR$. 
Since there is no bidding in $v$, the budgets of the players at $v$ remain unaffected after the transition.
For a fixed budget $B$, if $p_w$ is the probability of reaching $T$ in $i-1$ steps from the successor $w$ (of $v$), then the probability of reaching $T$ in $i$ steps from $v$ becomes $\sum_w p_w\cdot\ER(v)(w)$. 
(b)~Now suppose $v\in\VC$.
For a fixed probability $p$, we can ask for the least budget needed to reach $T$ from $v$.
If $B_+$ and $B_-$ are the maximum and minimum budgets required from any successor to reach $T$ in $i-1$ steps with probability $p$, then the budget required at $v$ for $i$-step reachability is $(B_++B_-)/2$.
This follows from the fact that the bid $(B_+ - B_-)/2$ will either lead to won bidding and budget $B_-$ or lost bidding and budget $B_+$.

The value iteration algorithm is now formally presented below.
\begin{align}
        &\rval^0_v\coloneqq
 \begin{cases}
 [0,1]^2 & \text{if } v\in T, \\
 [0,1]\times\{0\} \cup \set{1}\times \left[0,1\right] & \text{otherwise,}
 \end{cases}\\
        &\sval^0_v\coloneqq
 \begin{cases}
 [0,1]\times\{1\} \cup \set{0} \times [0,1] & \text{if } v\in T, \\
 [0,1]^2 & \text{otherwise,}
 \end{cases}
 \end{align}
 and for $i>0$, $\rval^i_v \coloneqq \bellman_v\left(\set{\rval^{i-1}_w\mid w\in\succ(v)}\right)$ and $\sval^i_v = \bellman_v\left(\set{\sval^{i-1}_w\mid w\in\succ(v)}\right)$, where the operator $\bellman_v$ is defined as follows:
 If $v\in\VR$,
 \begin{multline}
        \label{eq:vi-vr}
 \bellman_v\left(\set{\val_w\mid w\in\succ(v)}\right)\coloneqq\\
 \bigcup_{B \in [0, 1]} \left\{ \left(B, \sum_{w\in \succ(v)} \ER(v)(w) \cdot p_w\right)\, \middle|   \forall w\in\succ(v)\;.\; (B, p_w) \in \val_w \right\} 
 \end{multline}  
 and if $v\in \VC$,  
 \begin{equation}
    \label{eq:vi-vc}
 \begin{aligned}
        &\bellman_v(\set{\val_w \mid w\in\succ(v)})\coloneqq
 \bigcup_{p\in [0,1]} \left\{ \left(\frac{B_+ + B_-}{2}, p\right)\, \middle|\right.\\ 
                    &\ \ \ B_- \in \set{ B\mid\exists w\in\succ(v)\;.\; (B,p)\in \val_w},\\
                    &\left.\ \ B_+ \in \set{B\mid \forall w\in\succ(v)\;.\;(B,p)\in \val_w } \right\}.
 \end{aligned}
 \end{equation}
Fig.~\ref{fig:bellman operator illustration} illustrates the $\bellman_v$ operator in action.
Intuitively, the $\bellman_v$ operator averages the value sets of the successors along the $p$ axis for random vertices and along the $B$ axis for control vertices. 
Fig.~\ref{fig:value iteration illustration} illustrates the value sets computed using the value iteration algorithm.

	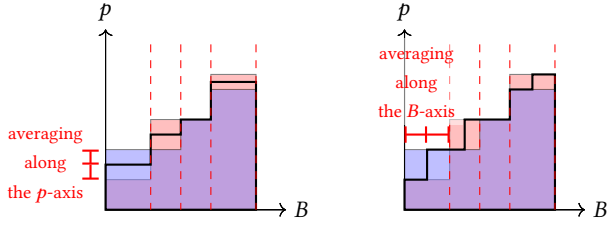
\begin{figure}
		\begin{tikzpicture}[scale=2]
			\draw[->]	(0,0)	--	(0,1.2)	node[above]	{$p$};
			\draw[->]	(0,0)	--	(1.2,0)	node[right]	{$B$};
			\draw[fill=red!50!white,opacity=0.5]	(0,0)	--	(0,0.2)	--	(0.3,0.2)	--	(0.3,0.6)	--	(0.7,0.6)	--	(0.7,0.9)	--	(1,0.9)	--	(1,0)	--	cycle;
			
				\draw[->]	(0,0)	--	(0,1.2);
				\draw[->]	(0,0)	--	(1.2,0);
				\draw[fill=blue!50!white,opacity=0.5]	(0,0)	--	(0,0.4)	--	(0.5,0.4)	--	(0.5,0.6)	--	(0.7,0.6)	--	(0.7,0.8)	--	(1,0.8)	--	(1,0)	--	cycle;
			
				\draw[->]	(0,0)	--	(0,1.2);
				\draw[->]	(0,0)	--	(1.2,0);
				\draw[thick]	(0,0)	--	(0,0.3)	--	(0.3,0.3)	--	(0.3,0.5)	--	(0.5,0.5)	--	(0.5,0.6)	--	(0.7,0.6)	--	(0.7,0.85)	--	(1,0.85)	--	(1,0)	--	cycle;
			
			\draw[dashed,red]	(0.3,1.1)	--	(0.3,0);
			\draw[dashed,red]	(0.5,1.1)	--	(0.5,0);
			\draw[dashed,red]	(0.7,1.1)	--	(0.7,0);
			\draw[dashed,red]	(1.0,1.1)	--	(1.0,0);
			
			\draw[red,|-|,thick]		(-0.1,0.3)	--	(-0.1,0.4);
			\draw[red,-|,thick]		(-0.1,0.3)	--	(-0.1,0.2);
			\node[align=center,red]		at		(-0.4,0.3)		{{\footnotesize averaging}\\ {\footnotesize along}\\{\footnotesize the $p$-axis}};
			
		\end{tikzpicture}
		\qquad
		\begin{tikzpicture}[scale=2]
		
			\draw[->]	(0,0)	--	(0,1.2)	node[above]	{$p$};
			\draw[->]	(0,0)	--	(1.2,0)	node[right]	{$B$};
			\draw[fill=red!50!white,opacity=0.5]	(0,0)	--	(0,0.2)	--	(0.3,0.2)	--	(0.3,0.6)	--	(0.7,0.6)	--	(0.7,0.9)	--	(1,0.9)	--	(1,0)	--	cycle;
			
				\draw[->]	(0,0)	--	(0,1.2);
				\draw[->]	(0,0)	--	(1.2,0);
				\draw[fill=blue!50!white,opacity=0.5]	(0,0)	--	(0,0.4)	--	(0.5,0.4)	--	(0.5,0.6)	--	(0.7,0.6)	--	(0.7,0.8)	--	(1,0.8)	--	(1,0)	--	cycle;
			
				\draw[->]	(0,0)	--	(0,1.2);
				\draw[->]	(0,0)	--	(1.2,0);
				\draw[thick]	(0,0)	--	(0,0.2)	--	(0.15,0.2)	--	(0.15,0.4)	--	(0.4,0.4)	--	(0.4,0.6)	--	(0.7,0.6)	--	(0.7,0.8)	--	(0.85,0.8)	--	(0.85,0.9)	--	(1,0.9)	--	(1,0)	--	cycle;
			
			\draw[dashed,red]	(0.3,1.1)	--	(0.3,0);
			\draw[dashed,red]	(0.5,1.1)	--	(0.5,0);
			\draw[dashed,red]	(0.7,1.1)	--	(0.7,0);
			\draw[dashed,red]	(1.0,1.1)	--	(1.0,0);
			
			\draw[red,|-|,thick]		(0,0.5)	--	(0.15,0.5);
			\draw[red,-|,thick]		(0.15,0.5)	--	(0.3,0.5);
			\node[align=center,red,fill=white,fill opacity=0.5,text opacity=1]		at		(0.08,0.85)		{{\footnotesize averaging}\\ {\footnotesize along}\\{\footnotesize the $B$-axis}};
		\end{tikzpicture}
		\caption{Illustration of the $\bellman_v$ operator for when $v\in \VR$ (left) and $v\in\VC$ (right). In both cases, we assume there are two successors whose values from the $(i-1)$-th iteration are given as the red and blue regions. 
		The outputs of $\bellman_v$ is shown as the set with thick boundaries.		
		For $v\in\VR$ (left), we assume uniform transition probabilities (i.e., $0.5$ for each successor).
		}
		\label{fig:bellman operator illustration}
	\end{figure}

The following theorem formally states the soundness of the above procedure for deciding which player has a winning policy.
Our proof is constructive, and shows how to obtain the corresponding winning policies.

\begin{theorem}\label{thm:bounded value iteration}
 Let $\tup{M,v,T,B,p}$ be a problem instance, $i\in \N$, and $\overline{B}\in [0,1]$.
 The following hold:
 \begin{enumerate}[(A)]
        \item $(\overline{B},p)\in\rval_{v}^i$, $B > \overline{B} \Rightarrow \winpolZ(B,p,v,\reach^{M,i}(T))\neq \emptyset$,
        \item $(\overline{B},p)\in\sval_{v}^i$, $B < \overline{B} \Rightarrow \winpolO(B,p,v,\reach^{M,i}(T))\neq \emptyset$.
 \end{enumerate}
\end{theorem}

\begin{proof}
 We simultaneously prove both (A) and (B) hence we use $\val$ to denote either $\rval$ or $\sval$.
 The winning policy $\pol$ can be extracted from the computed values $\val^i_v$ inductively.

 Let $(\overline{B},p)\in \val_v^i$ and suppose the reachability player's initial budget is $B = \overline{B} + s$ with $s>0$ for $\rval$ and $s<0$  for $\sval$.
    
 If $i=0$, the claim is trivially true since regardless of the chosen policy,
 any budget is sufficient to guarantee reaching $T$ in zero steps if $v \in T$, and
 any budget is sufficient to reach it with probability zero otherwise.
 Respectively, any budget is sufficient to avoid $T$ with probability one if $v \not\in T$, and
 any budget is sufficient to avoid it with probability zero otherwise.

 Otherwise, assume $i > 0$ and first discuss the case when $v\in\VC$.
 In the first step, $\pol$ identifies $B_-$ and $B_+$ as defined in \eqref{eq:vi-vc} such that $(B_- + B_+)/2 = \overline{B}$,
 bids $\lvert B_+-B_-\rvert/2$, and tries to move the token to the successor $w$ for which $(B_-,p)\in \val_w^{i-1}$.
 If $\pol$ wins the bidding, it moves the token to $w$ and gives the budget to the opponent yielding the new budget $B' = B_- + s$ for the reachability player.
 Otherwise, the opponent moves the token to any successor $w$ yielding the new budget $B' = B_+ + s$.
 By definition of $B_+$, $(B_+,p)\in \val_{w}^{i-1}$ for any choice of $w$.
 Therefore, regardless of the new vertex $w$, $\pol$ has enough budget to continue according to a policy in $\winpol{j}(B', p, w, \reach^{M,i-1}(T))$
 which is nonempty by the induction hypothesis (here $j$ is $\mathsf{R}$ if $\val$ is $\rval$ and $\mathsf{S}$ if $\val$ is $\sval$).

 Now consider the case when $v\in\VR$. By the inductive definition of $\val_v^i$ in \eqref{eq:vi-vr},
 there exists $p_w$ for each $w\in\succ(v)$ such that $\sum_{w\in\succ(v)} \ER(v)(w)\cdot p_w = p$ and $(B,p_w)\in \val_w^{i-1}$.
 By the inductive hypothesis, $\winpol{j}(B, p_w, w, \reach^{M,i-1}(T))$ contains a policy $\pol_w$ for each $w\in\succ(v)$,
 hence $\pol$ simply follows $\pol_w$ regardless of the stochastic outcome.
\end{proof}

The computability of the values at each step follows from their finite representations.
In particular, the sets $\rval_v^i$ and $\sval_v^i$ have a ``staircase form'' (see Fig.~\ref{fig:bellman operator illustration})
and can be represented by the
corner points of the steps.
Before formalizing this, define the order $\ord$ as $(B, p) \ord (B', p')$ if and only if $B \geq B'$ and $p \leq p'$.
A $\ord$-downward (or $\ord$-upward) closure of a set $S$ is the set of all points $(B, p)$ such that there exists $(B', p')\in S$ with $(B, p)\ord (B', p')$ (or $(B', p')\ord (B, p)$).

\begin{lemma}
    \label{lem:fin-rep}
 Let $M$ be an MDP, $T$ be target vertices, and $v$ be a vertex in $M$.
 For every $i \in \N$, there exists
 a finite set $G \subseteq [0,1]^2$ of at most $3|V|^i$ points such that
    $\rval_v^i$ is the $\ord$-downward closure of $G$, and $\sval_v^i$ is the $\ord$-upward closure of $G$.
 Moreover, all boundary points of $\rval_v^i$ belong to $\sval_v^i$, and vice versa.
\end{lemma}

The sets $\rval_v^0$ and $\sval_v^0$ are downward and upward closures of $\set{(0,0), (0,1), (1,1)}$ or $\set{(0,0), (1,0), (1,1)}$, depending on whether $v$ is in $T$ or not.
Thus they indeed have the staircase shape with a single step. To prove the Lemma, it can be shown that the staircase shape propagates through the operator $\bellman_v$, which we formally prove  in the supplementary material.

A direct consequence of Lem.~\ref{lem:fin-rep} is that the sets $\rval_v^i$ and $\sval_v^i$ can be computed in exponential time and space for every $i$.
\begin{corollary}
 The sets $\rval^i_v$ and $\sval^i_v$ for each $v$ and $i$ can be computed in $\mathcal{O}(|V|^i)$ time and $\mathcal{O}(|V|^i)$ space.
\end{corollary}
Another consequence is determinacy, i.e., one of the players always fulfills the respective bounded-horizon specification.

\begin{corollary}\label{thm:bounded:determinacy}
 Let $\tup{M,v,T,B,p}$ be an arbitrary problem instance.
 For every $i$, the point $(B,p)$ belongs to at least one of the sets $\rval_{v}^i$ and $\sval_{v}^i$.
\end{corollary}

The proof can be found in the supplementary material.

\section{From Bounded to Unbounded Horizon}
\label{sec:value iteration}

\subsection{Limiting Behavior of Value Iteration}
 If we continue the bounded-horizon value iteration for increasing horizon, we obtain the following values in the limit:
    $$
 \optcrv_v \coloneqq \closure{\bigcup_{i=0}^\infty \rval_v^i} \text{ and }\coptcrv_v \coloneqq \bigcap_{i=0}^\infty \sval_v^i,
    $$
 where $\closure{S}$ denotes the closure of a set $S$ in the Euclidean metric; note that $\coptcrv_v$ is closed by construction.
 In this section, we establish a connection between $\optcrv_v,\coptcrv_v$ and the true values that are winning for the respective player in the unbounded horizon setting.
 In particular, we relate the limit sets to the threshold $\thres_v$, and present an algorithm for Prob.~\ref{prob:exact} that runs in doubly exponential time.
Our proof also constructs the winning policy for the reachability player, if one exists; the construction of the safety player's winning policy remains open.
 In the following, we will use the notation $\interior{S}$ to denote the interior of the set $S$.

 \begin{theorem}\label{thm:limiting values}
 	Let $\tup{M,v,T,B,p}$ be a problem instance.
 	The following hold:
 	\begin{enumerate}[(A)]
 		\item $(B,p)\in \interior{\rval_{v}^*} \Rightarrow \winpolZ(B,p,v,\reach^{M}(T))\neq \emptyset$,
 		\item $(B,p)\in \interior{\sval_{v}^*} \Rightarrow \winpolO(B,p,v,\reach^{M}(T))\neq \emptyset$.
 	\end{enumerate}
 \end{theorem}
    
 \begin{proof}
 Let $(B, p) \in \interior{\optcrv_v}$. Then it belongs to the interior of $\rval_v^i$ already for some $i \in \N$.
 Every reachability policy winning on a finite horizon is, in particular, winning on the infinite horizon,
 hence Thm.~\ref{thm:bounded value iteration} implies that $\winpolZ(B, p, v, \reach^M(T))$ is non-empty.

 Now suppose $(B, p)$ is in the interior of $\coptcrv_v$
 implying there also exists $\overline{B} = B + s$ for some $s > 0$ such that $(\overline{B}, p) \in \coptcrv_v$.
 We describe the winning policy $\polO$ of the safety player.
 Besides the budget, the policy maintains a requested probability of avoiding $T$, initially set to $p$.
 It can be shown (see the supplementary material)
 that the sets $\coptcrv_v$ form a fixpoint of the operator $\bellman$.
 Thus whenever in $v \in \VC$, the safety player can determine $B_+$ and $B_-$ from equation \eqref{eq:vi-vc} such that $(B_- + B_+)/2 = \overline{B}$
 and bid $(B_--B_+)/2$. Upon winning (thus increasing the reachability player's budget to $B_--s$) she moves the token to the successor $w$ for which $(B_-,p)\in \coptcrv_w$. 
 In the case of losing the bid (thus decreasing the budget to $B_+-s$), the reachability player chooses any successor $w$ and by definition of $B_+$, we know $(B_+,p)\in \coptcrv_w$.
 Whenever $v \in \VR$, the safety player determines $p_w$ for each $w\in\succ(v)$ such that $\sum_{w\in\succ(v)} \ER(v)(w)\cdot p_w = p$ and $(B,p_w)\in \coptcrv_w$, and
 upon moving to a successor $w$, she updates the requested probability to $p_w$.

 Note that after moving to $w$, the respective point $(B', p')$ --- equal to either $(B_- -s, p), (B_+ -s, p)$, or $(B-s, p_w)$ --- again lies in $\coptcrv_w$.
 Hence the two above rules can be applied ad infinitum, maintaining the invariant that $(B'+s, p') \in \coptcrv_w$.
 Since $(B',p') \in \coptcrv_w$ implies $(B',p') \in \sval_w^i$ for every $i$, policy $\polO$ essentially coincides with
 the policy described in the proof of Thm.~\ref{thm:bounded value iteration} for any finite horizon length $i$.
 Therefore, $\polO$ belongs to $\winpolO(B, p, v, \reach^{M, i}(T))$ for every $i$.
 We show by contradiction that $\polO \in \winpolO(B, p, v, \reach^M(T))$.
 If there was $\polZ$ that guarantess probability of reaching $T$ higher than $p$, i.e., $\pr{v, B}{\polZ,\polO}(\reach^M(T)) = q > p$,
 then it would achieve almost $q$ already in finite number of steps, i.e., $q \geq \pr{v, B}{\polZ,\polO}(\reach^{M, k}(T)) > p$ for some $k$.
This contradicts the fact that $\polO$ is in $\winpolO(B, p, v, \reach^{M, k}(T))$, hence there is no such $\polZ$, and $\polO$ is winning for the safety player.
\end{proof}

We now characterize the threshold and show that it is completely separating, thereby establishing determinacy.
The proof can be found in the supplementary material.

\begin{corollary}[Determinacy]\label{cor:determinacy}
   For every MDP $M$, target vertices $T$, and vertex $v$ in $M$,
   it holds that $\optcrv_v \cap \coptcrv_v=\thres_v$.
   Moreover, the threshold is completely separating.
\end{corollary}

\subsection{On the Decidability of the Quantitative Reachability Problem}
From Thm.~\ref{thm:limiting values}, it follows that Problem~\ref{prob:exact} reduces to the membership problem of deciding which of the sets $\optcrv_{v}$ and $\coptcrv_{v}$ contain the given pair $(B,p)$.
 Unfortunately, the decidability of this question remains open.
 As far as the value iteration is concerned, the difficulty with decidability comes from the situation when $(B,p)\in\thres_v$, because the true $\thres_v$ is obtained only in the limit.
 Therefore we are not guaranteed to decide the membership of $(B,p)$ in finite time by only using value iteration.
  To circumvent this ``edge case,'' we make the following assumption.

 \begin{assumption}\label{ass:approx}
 Let $\tup{M,v,T,B,p}$ be a problem instance.
 The given pair $(B,p)$ does not belong to $\thres_v$.
 \end{assumption}    
 In other words, we assume that the pair $(B,p)$ lies either in the \emph{interior} of $\rval_v^*$ or in the \emph{interior} of $\sval_v^*$.
 For every point $(B,p)\in\interior{\rval_v^*}$, Lem.~\ref{lem:convergence speed} below provides an upper bound on the iteration index after which $(B,p)$ will be included inside $\rval_v^i$; and excluded from $\sval_v^i$, respectively. 

We first introduce some notation.
We will use $\ER_{\min}$ to denote the minimum transition probability in $M$; if all transitions are non-probabilistic, we will set $\ER_{\min}=\frac{1}{2}$.
We further denote by $d_\infty(x, y)$ and $d_\infty(x, Y)$ the $L_\infty$ distance of a point $x$ to another point $y$ and to the set $Y$.
Of importance is the situation when either $x\in \rval_v^*$ and $Y=\sval_v^*$, or $x\in \sval_v^*$ and $Y=\rval_v^*$, in which case $d_\infty(x,Y)$ is called the \emph{distance of $x$ from $v$'s threshold}.
Given the sets $X, Y\subseteq [0,1]^2$, the Hausdorff distance $\dH(X,Y)$ between $X$ and $Y$ is the largest distance one needs to travel starting from any point in either $X$ or $Y$ to reach the closest point in the other set.
Formally,
\begin{align*}
 \dH(X,Y)\coloneqq \max\set{\max_{y\in Y}\min_{x\in X}\lVert x-y\rVert_{\infty}, \max_{x\in X}\min_{y\in Y}\lVert x-y\rVert_{\infty}}.
\end{align*}  

It can be shown that there exist optimal policies that admit a short path to $T$ regardless of the game history,
and that after exponentially many random choices, such a path is traversed with sufficiently high probability.
This idea yields the following lemma.
\begin{lemma}\label{lem:convergence speed}
 Let $\tup{M,v,T,B,p}$ be a problem instance and
 suppose $(B,p)$ is in the interior of $\rval_{v}^*$ with its distance from $v$'s threshold being $\epsilon>0$.
 Then for every $n_\epsilon$ such that
    $$n_\epsilon \geq 4|\V|\log\left(2/\epsilon\right)\ER_{\min}^{-2|\V|},$$
  it holds that $(B,p)\in \rval_{v}^{n_\epsilon}$.
\end{lemma}
Full proof of the lemma is provided in the supplementary material.
Lem.~\ref{lem:convergence speed} paves the way to the following algorithm for Prob.~\ref{prob:exact} under Assump.~\ref{ass:approx}.

\noindent\hrulefill\\
\noindent\textbf{Algorithm: \AlgExact}
\\
For suitable $c, d>0$ from Lemma \ref{lem:convergence speed}, repeat for $i=0,1,2,\ldots$:
\begin{enumerate}[(1)]
    \item If $(B,p)\in \rval_{v}^i$, return ``$(B,p)\in \interior{\rval_{v}^*}$''.
    \item If $d_\infty((B,p), \rval_{v}^i)> 2e^{{-i\ER_{\min}^{2|\V|}}/{4|\V|}}$, return ``$(B,p)\in \interior{\sval_{v}^*}$.''
\end{enumerate}
\noindent\hrulefill\\
It is easy to see that \AlgExact is sound:
If it returns that $(B,p)\in \rval_{v}^*$, i.e., if (1) happens before (2), then soundness follows from the fact that $\rval_{v}^i\subseteq \rval_{v}^*$ for each $i$.
If it returns that $(B,p)\in \sval_{v}^*$, i.e., if (2) happens before (1), then it follows from Lemma~\ref{lem:convergence speed} that $(B,p)$ is farther away from $\rval_{v}^i$ than the threshold of $v$.
Therefore, $\rval_{v}^*$ does not contain $(B,p)$, and from Cor.~\ref{cor:determinacy}, it follows that $\sval_{v}^*$ must contain $(B,p)$.
Despite soundness, \AlgExact is a semi-decision procedure, because, in general, neither (1) nor (2) is guaranteed to be triggered after finitely many $i$, as has been explained earlier.
Luckily, termination is guaranteed when Assump.~\ref{ass:approx} holds, in which case \AlgExact is sound and complete.
We formally state these claims below.

\begin{theorem}
   \label{thm:exact algorithm}
 \AlgExact is a semi-decision procedure for Prob.~\ref{prob:exact}, and is a sound and complete algorithm when Assump.~\ref{ass:approx} is fulfilled by the given problem instance.
 For the latter case, if the distance between $(B,p)$ and $v$'s threshold is $\epsilon>0$, then \AlgExact terminates in at most 
 $\mathcal{O}\left(\log\left(1/\epsilon\right)|\V|\ER_{\min}^{-2|\V|}\right)$
 iterations, yielding the space and time complexity
 $$|\V|^{\mathcal{O}\left(\log\left(1/\epsilon\right)|\V|\ER_{\min}^{-2|\V|}\right)}.$$
\end{theorem}

\subsection{Abstraction for Complexity Reduction}

To reduce the high computational complexity of \AlgExact, we present an approximation algorithm that is sound and complete
in the sense of Thm. \ref{thm:exact algorithm} although its running time is only exponential in $|\V|$.
The high complexity of \AlgExact stems from the arbitrarily high precision of the value set representations.
The idea of our algorithm is to discretize the two dimensions---budget and probability---using a fixed finite precision.
The approximate value iteration begins by rounding off $\rval_v^0$ or $\sval_v^0$  (whether rounding up or down depends on the dimension and the value set) to the closest discrete level along each dimension.
At each subsequent step, first the $\bellman_v$ operator (defined in \eqref{eq:vi-vc} and \eqref{eq:vi-vr}) is applied on the approximate values from the previous iteration.
Since $\bellman_v$ may produce value sets with higher precision than the chosen one,
the obtained sets are further 
rounded off again to the closest discrete level along each dimension.
We show that the procedure yields a sufficiently low approximation error that depends on the number of iterations and the chosen precision.
We formalize the abstract safety value iteration below; the abstract reachability values are later obtained through duality.

Let $\alpha\in [0,1]$ be a constant parameter, called the \emph{grid size}, such that $1/\alpha$ is an integer.
We define a uniform rectangular grid $X$ on the value space $[0,1]^2$ such that the space is divided in squares of length $\alpha$, i.e., $$X\coloneqq \set{ [p\alpha,(p+1)\alpha]\times [q\alpha,(q+1)\alpha]\mid 0\leq p,q < 1/\alpha  }.$$
Given a relation $S\subseteq [0,1]^2$ over the budgets and probabilities, define $\grid(S)$ as the \emph{over-approximation} of $S$ using the elements of $X$, i.e., $\grid(S) \coloneqq \bigcup \set{x\in X\mid x \cap S \neq \emptyset} $.

For every $v\in\V$, the abstract 
safety value iteration computes the decreasing (with respect to ``$\subseteq$'') sequence $\abs{\sval}_v^0,\abs{\sval}_v^1,\ldots$ 
defined as below:
\begin{align*}
    &\abs{\sval}_v^0 \coloneqq \grid(\sval_v^0),\\
 \forall i>0:\quad &\abs{\sval}_v^i \coloneqq \grid\circ \bellman_v\left(\set{\abs{\sval}_w^{i-1}\mid w\in\succ(v)}\right).
\end{align*}
Henceforth, we will refer to $\abs{\sval}_v^i$ 
as \emph{abstract} values, and $\sval_v^i$
as \emph{concrete} values.
The following lemma shows that the abstract values always over-approximate the concrete values, and the approximation error remain bounded by $\alpha (i+1)$ for every $i$.
The proof is provided in the supplementary material.

\begin{lemma}
   \label{lem:abstraction precision}
 For every $v\in\V$ and every $i$, $\abs{\sval}_v^i \supseteq \sval_v^i$, and, moreover, $\dH(\abs{\sval}_v^i, \sval_v^i)\leq \alpha (i+1)$.
 The abstract value $\abs{\sval}_v^i$ can be computed in $2|\V|^2i/\alpha$ time and $2|\V|/\alpha$ space.
\end{lemma}

Given an abstract safety value, we can define an abstract reachability value $\abs{\rval}_v^i \coloneqq [0,1]^2 \setminus \abs{\rval}_v^i \cup \rval_v^0$.
Note that $\abs{\rval}_v^i$ is under-approximation of $\rval_v^i$ since both, the complement of $\abs{\sval}_v^i$ and $\rval_v^0$, are subsets of $\rval_v^i$.
Moreover, whenever $d_h\left(\abs{\sval}_v^i, \sval_v^i\right) \leq \overline{\epsilon}$ then 
$d_h\left(\abs{\rval}_v^i, \rval_v^i\right) \leq \overline{\epsilon}$.

We now present \AlgApprox, an improved algorithm for Prob.~\ref{prob:exact}.
\AlgApprox uses the same principle as \AlgExact, but instead of using the exact value iteration, uses the approximate counterpart along with an iterated refinement of the resolution of the grid.

\noindent\hrulefill\\
\noindent\textbf{Algorithm: \AlgApprox}\\
For suitable $c, d>0$ from Lemma \ref{lem:convergence speed}, repeat for $h = 0,1,2,\ldots$:
\begin{enumerate}
   \item Set $\overline{\epsilon} \coloneqq 2^{-h}$, $n \coloneqq \lceil\log\left(\frac{1}{\overline{\epsilon}}\right)\left(c + d\ER_{\min}^{-2|\V|}\right)\rceil$, $\alpha \coloneqq \overline{\epsilon}/n$.
   \item Compute $\abs{\sval}_{v}^n$ for a given $\alpha$ and its complement $\abs{\rval}_{v}^n$.
   \item If $(B,p) \in \abs{\rval}_{v}^n$, return ``$(B,p)$ is in $\interior{\rval_{v}^*}$.''
   \item If $d_\infty\left((B,p), \abs{\rval}_{v}^n\right) \geq 2\overline{\epsilon}$, return ``$(B,p)$ is in $\interior{\sval_{v}^*}$.''
\end{enumerate}
\noindent\hrulefill\\

The algorithm \AlgApprox iteratively looks for a precision $\overline{\epsilon}$ that is sufficient to decide whether $(B,p)$ is in $\rval_{v}^*$ or $\sval_{v}^*$.
Since $\abs{\rval}_v^n \subseteq \rval_v^n \subseteq \rval_v^*$, every time the algorithm stops by Step~(3), the decision is sound.
For a given precision $\overline{\epsilon}$, the algorithm computes enough steps of the abstract value iteration on a fine-enough grid to ensure that:
a)~the abstract value $\abs{\rval_v^n}$ is at most $\overline{\epsilon}$-far from $\rval_v^n$, and
b)~the concrete value $\rval_v^n$ is at most $\overline{\epsilon}$-far from $\optcrv_{v}$.
Hence the decision on line (4) is also sound, by the triangular inequality. Finally, provided that $(B,p)$ is in $\epsilon$-distance from the threshold,
the algorithm will eventually set $\overline{\epsilon} \leq \frac{\epsilon}{2}$, decide by (3) or (4), and terminate.

\begin{theorem}
   The restriction of Prob.~\ref{prob:exact} to inputs satisfying Assump.~\ref{ass:approx} lies in EXPTIME. In particular,
   \AlgApprox is a sound and complete algorithm when Assump.~\ref{ass:approx} is fulfilled by the given problem instance,
   and provided the distance between $(B,p)$ and $v$'s threshold is $\epsilon>0$, then \AlgApprox terminates in at most
   $h = \lceil \log(1/\epsilon) \rceil + 1$ iterations
   yielding the time complexity
   $$\mathcal{O}\left(\frac{|\V|^2}{\epsilon}\log\left({1/\epsilon}\right)^3\ER_{\min}^{-4|\V|}\right).$$
\end{theorem}

\subsection{Lower Complexity Bounds}

We show that the exact problem for general MDPs is at least as hard as computing values in simple stochastic games (SSG), which is known to be in NP\,$\cap$\,co-NP and whether it belongs to P remains to be a long-standing open problem.

SSGs are generalizations of MDPs where the control vertices are partitioned into two sets based on their (fixed) ownerships by two players, referred to as \PZ and \PO.
Formally, an SSG is a tuple $\tup{\V,\VZ,\VO,\VR,\EC,\ER}$, where $\VZ$ and $\VO$ are the vertices owned by \PZ and \PO respectively, and $\V$, $\VR$, $\EC$, and $\ER$ are adapted from their definitions in MDPs as follows: $\V=\VZ\cup \VO\cup \VR$, $\EC\colon \VZ\cup \VO\to\VR$, and $\ER\colon \VR\to \distr(\VZ\cup \VO)$.
Without loss of generality, we assume every SSG has a designated initial vertex which is in $\VO$ and every path belongs to $(\VO\VR\VZ\VR)^\omega$,
i.e., each kind of vertices evenly alternate (achievable by polynomial blowup).

A \PI \emph{policy} is a function that maps every vertex $v \in \V_j$ to one of $v$'s successors.
Suppose, \PZ is the reachability player trying to reach a given target $T_G\subseteq \V$, and \PO is the safety player trying to avoid it.
The \emph{value} of a vertex $v$ is the maximum probability $p$ such that the reachability player can reach $T_G$ with probability at least $p$ against any policy of the safety player.
We refer to Condon~\cite{Con90} for detailed formal definitions.

\begin{theorem}\label{thm:lower bounds}
 For any SSG $G$ with initial vertex $v$, there exists a bidding game $G_B$ whose size is polynomial in the size of $G$, and such that the value of the reachability objective in $G$
 is equal to the minimal $p$ that satisfies $(\frac{1}{3}, p) \in \coptcrv_v$ in $G_B$.
\end{theorem}

\begin{proof}
Let $G$ be an SSG with vertices $\V_G$.
The structure of $G_B$ is the same as that of $G$
except that the control vertices have an extra edge appended.
Namely, for every $v \in \VZ$, there is a new (non-target) sink $v'$ with an edge $(v,v')$,
and for every $v \in \VO$, there is a new target sink $v'$ with an edge $(v,v')$.
We set $\VC$ to be the union of $\VZ$ and $\VO$, and $T$ to be $T_G$ together with the new targets.

Suppose the value of $G$ at the initial vertex $v$ is $p$, and
let $q$ be minimal such that $\left(1/3, q\right) \in \coptcrv_v$.
We first show that $q \leq p$. We will prove that $\left(1/3-s, p\right)$ belongs to $\coptcrv_v$ for every $s>0$,
since this implies $\left(1/3, p\right) \in \coptcrv_v$ as $\coptcrv_v$ is closed.
By definition of $p$, there exists a policy $\pi_1$ of \PO in $G$ that guarantees avoiding $T_G$ with probability at least $1-p$ against any policy of \PZ.
We use $\pi_1$ to construct a policy $\polO$ in $G_B$.
The policy $\polO$ is defined as follows:
For every $w\in \VO$, $\polO(B, w) = (1/3, \pi_1(w))$, and for every $w \in \VZ$, $\polO(B, w) = (1/3, w')$ where $w'$ is the new sink.
Note that since the initial budget of the reachability player is less than $1/3$ and the control vertices alternate evenly,
his budget will be less than $1/3$ on every visit of $\VO$.
Hence the reachability player loses every bidding in $\VO$, regardless of his policy.
Therefore, for every $\polZ$, the game ends in a new non-target sink or forever stays in $\V_G$,
where $\polO$ playes as $\pi_1$ in $G$. Hence the token avoids $T$ with probability at least $1-p$.

We next show that for every $p' < p$,
there exists $s>0$ such that $\left(1/3-s, p'\right) \in \optcrv_v$.
Since $\optcrv_v$ is $\ord$-downward closed,
this implies $(1/3, p'') \in \interior{\optcrv}$ for every $p'' < p'$.
Since $p'$ ranges over all reals less than $p$, we derive $(1/3, p'') \in \interior{\optcrv_s}$ for every $p''<p$, which in turn implies $q \geq p$.
By definition of $p$, there exists a policy $\pi_0$ of \PZ in $G$ that reaches $T$ with probability at least $p$ against any policy $\pol_1$.
For every $p' < p$, we use $\pi_0$ to find the respective $s$ and construct a policy $\polZ$ that belongs to $\winpolZ(1/3-s, p')$.
In $\VO$, the policy $\polZ$ bids all-in, since winning the bid directly leads to $T$.
I.e. $\polZ(B, w) = (B, w')$ for every $w \in \VO$, where $w'$ is the new target.
In $\VZ$, if the current budget is less than $1/2$, the safety player can bid all-in and win the game, thus the bid of $\polZ$ is irrelevant.
Otherwise, since the reachability player wins the ties, it is enough to bid the opponent's budget $1-B$.
Hence, $\polZ(B, w) = (\max\set{0,1-B}, \pi_0(w))$ for every $w \in \VZ$.
Note that whenever the budget of the reachability player is $2/3 - \theta$ in $\VZ$ for some $\theta > 0$,
it will be at least $2/3 - 4\theta$ when the token reaches $\VZ$ the next time, due to the even alternation of the control vertices.
Hence for any initial budget $1/3-s$ and any $n$ such that $1/2 \leq 2/3 - 4^ns$, the reachability player can win first $n$ biddings in $\VZ$.
In partucular, for any safety player policy, the game will stay in $\VZ \cup \VO$ at least for $4n$ steps or deviate to $T$ earlier.
Since in $G$, $\pi_0$ reaches $T_G$ with probability $p'$ already in finite number of steps,
the policy $\polZ$ reaches $T$ with probability $p'$ in $G_B$ provided $s$ is small enough.
\end{proof}

%
\section{Bidding Games on Restricted MDP-s}

Now we consider the special cases of Prob.~\ref{prob:exact} for acyclic and tree-shaped MDPs, which are MDPs whose underlying transition graphs are acyclic (with loops on sinks) and rooted tree, respectively.

\subsection{Acyclic MDPs}\label{sec: acyclicMDP}

It is easy to see that, contrary to general MDPs (with cycles), the value iteration algorithm (from Sec.~\ref{sec:value iteration}) converges in at most $|\V|$ iterations for acyclic MDPs. 
This implies that \AlgExact will always terminate in at most $|V|$ iterations.

\begin{restatable}{lemma}{valueiterforDAG}
	\label{lemma: valueiterforDAG}
	For acyclic MDPs, for every vertex \(v\), \(\rval^{|\V|}_v = \optcrv_v\) and $\sval^{|\V|}_v = \coptcrv_v$.
\end{restatable}

Lem.~\ref{lemma: valueiterforDAG} implies that Prob.~\ref{prob:exact} is in EXPTIME for acyclic MDPs.

\begin{theorem}\label{thm:acyclic MDP}
	For acyclic MDPs, \AlgExact is a sound and complete algorithm for Prob.~\ref{prob:exact}
	and terminates in at most $|\V|$ iterations, yielding the space and time complexity $\mathcal{O}\left(|\V|^{|\V|}\right)$.
\end{theorem}

\subsection{Tree-Shaped MDPs}\label{sec:treeshapedMDP}

In case the MDP is tree-shaped, we can solve Prob.~\ref{prob:exact} in NP $\cap$ co-NP.
The main idea of the proof is to find a certificate
of the fact that $(B,p)$ belongs to $\rval^{|\V|}_v$. According to the inductive definition of $\rval^{|\V|}_v$ by \eqref{eq:vi-vr} and \eqref{eq:vi-vc},
the presence of $(B,p)$ in $\rval^{|\V|}_v$ can be witnessed by a point $(B_w, p_w) \in \rval^{|\V|-1}_w$ for every $w \in \succ(v)$.
A similar witness can be constructed to show that each of the points $(B_w, p_w)$ belongs to $\rval^{|\V|-1}_w$, and the process can be repeated recursively up to the base case $(B_u, p_u) \in \rval_u^0$.
Therefore, whenever $(B, p)$ belongs to  $\optcrv_v$, there exists a certificate of this fact in the form of a finite set of points satisfying the relations in \eqref{eq:vi-vr} and \eqref{eq:vi-vc}.
If $M$ is a tree, the certificate moreover contains a single point $(B_u, p_u)$ for every vertex $u \in \V$.
The whole certificate then satisfies the following constraints:
\begin{align*}
	B_v &\leq B,~p_v \geq p\\
	\substack{\forall u \in \VC \\ \exists u^-, u^+ \in \succ(u)}: ~B_{u} &= \frac{B_{u^-} + B_{u^+}}{2},~p_u \leq \!\!\!\!\!\!\!\min_{w \in \succ(u)} \!p_w, B_{u^+} \geq \!\!\!\!\!\!\!\max_{w \in \succ(u)} \!\!\!\!B_w\\
	\forall u \in \VR: ~B_u &\geq \!\!\!\!\!\max_{w \in \succ(u)} B_w,~p_u = \sum_{w\in \succ(u)} \ER(u)(w)\cdot p_w\\
	\forall t \in T:~0 \leq B_t &\leq 1,~0 \leq p_t \leq 1\\
	\forall z \in Z:~B_z &= 1 \text{ or }~p_z = 0,
\end{align*}
where \(Z\) is the set of leaves not in \(T\).
By fixing a choice of $u^-$ for every control vertex $u$, and a choice of which of the two equalities should hold for each $z \in Z$,
we create a concrete linear program. The point $(B, p)$ belongs to $\optcrv_v$ if and only if there is a choice that makes the linear program feasible.
The same idea can be used to prove a point belongs to $\coptcrv_v$. The following result follows.

\begin{restatable}{theorem}{treeshapedMDP}
	\label{thm: treeshapedMDP}
	For tree-shaped MDPs, Prob.~\ref{prob:exact} is in NP \(\cap\) co-NP.
\end{restatable}

\section{Conclusions and Future Work}
We studied bidding games on MDPs with quantitative reachability and safety specifications.
We show that thresholds are binary relations over budgets and probabilities.
This makes their computation significantly more challenging than traditional bidding games on graphs, for which thresholds are scalars (budgets).
We developed a new value iteration algorithm for approximating the threshold up to arbitrary precision, and showed how it can be used to decide whether a given initial budget $B$ suffices to win w.p. at least $p$, assuming $(B,p)$ is not on the threshold (Assump.~\ref{ass:approx}). In acyclic and tree-shaped MDPs, Assump.~\ref{ass:approx} is not required and the decision procedure becomes significantly more efficient.
\stam{
We developed a new value iteration algorithm for approximating the threshold up to arbitrary precision, and showed how it can be used to solve Prob.~\ref{prob:exact}, i.e., deciding if a given initial budget $B$ is sufficient to satisfy the specification with a given probability $p$, under Assump.~\ref{ass:approx} (stating that $(B,p)$ is not on the threshold).
Additionally, the special cases of acyclic and tree-shaped MDPs are considered for which Assump.~\ref{ass:approx} is not required and the decision procedure for Prob.~\ref{prob:exact} becomes significantly more efficient.
}

A number of questions remain open: Is Prob.~\ref{prob:exact} decidable without Assump.~\ref{ass:approx}? What are the exact complexities? (There is a big gap between the upper and lower complexity bounds.)
Furthermore, several interesting extensions can be considered, namely extensions to richer classes of specifications (like $\omega$-regular and mean-payoff) and extensions to different forms of bidding mechanisms (like poorman and taxman, both with and without charging).
Another interesting question is the equivalence with stochastic models (recall that bidding games are equivalent to random-turn games). 
This is still unclear, because even if the threshold budget were simulated by random turn assignments, this randomness would not ``blend'' with the existing randomness (in the random transitions) in the MDP, and we would obtain stochastic games with two sources of probabilities, which have not been studied to the best of our knowledge.
Finally, the foundation of auction-based scheduling~\cite{AMS24} on MDPs is now ready, and it will be interesting to investigate how policy synthesis for multi-objective MDPs can benefit from it.

%


\begin{acks}
We thank Alon Krymgand for suggesting the proof technique used for proving Thm.~\ref{thm:lower bounds}, which simplified our original proof.
G.\ Avni and S.\ Sadhukhan were supported by the ISF grant no. 1679/21, M.\ Kure\v{c}ka and P.\ Novotn\'{y} were supported by the Czech Science Foundation grant no. GA23-06963S, and K.\ Mallik was supported by the ERC project ERC-2020-AdG 101020093. 
\end{acks}



\bibliographystyle{ACM-Reference-Format} 
\bibliography{ga,mo_references}

\pagebreak
\appendix

\section*{Supplementary Materials}

\section{Additional Proofs}

\subsection{Bounded-Horizon Value Iteration}

\newtheorem*{L3}{Lemma \ref{lem:fin-rep}}

\begin{L3}
 Let $M$ be an MDP, $T$ target vertices, and $v$ be a vertex in $M$.
 For every $i \in \N$, there exists
 a finite set $G \subseteq [0,1]^2$ of at most $3|V|^i$ points such that
    $\rval_v^i$ is the $\ord$-downward closure of $G$, and $\sval_v^i$ is the $\ord$-upward closure of $G$.
 Moreover, all boundary points of $\rval_v^i$ belong to $\sval_v^i$ and vice versa.
\end{L3}
\begin{proof}
 We prove the result by induction on $i$. For the base case $i = 0$, by definition of $\rval_v^0$ and $\sval_v^0$,
 the claim holds with $G = \set{(0, 0), (0, 1), (1,1)}$ or $G = \set{(0, 0), (1, 0), (1,1)}$ depending on whether $v\in T$ or not.
    
 For the induction step, suppose the claim holds for $i$.
 We only consider the case $v \in \VR$ since the proof for $v \in \VC$ is analogous.
 The definition of the operator in \eqref{eq:vi-vr} implies
 that whenever the sets $\rval_w^{i}$ and $\sval_w^{i}$ share a boundary, then the same holds for $\rval_v^{i+1}$ and $\sval_v^{i+1}$.
 Further, by the induction hypothesis, there exists the finite representation $G_w$ for each $w\in\succ(v)$.
 Let $(B, p) \in \rval_v^{i+1}$. From the $\ord$-downward closedness of $\rval_w^{i}$ and equation \eqref{eq:vi-vr} it follows that there exists
    $B'\geq B$ such that $(B', p) \in \rval_v^{i+1}$, and $(B',q) \in G_w$ for some $q$ and $w\in\succ(v)$.
 Hence $\rval_v^{i+1}$ is fully determined by the set of all its boundary points whose $B$-coordinates appear in some $G_w$.
 A similar argument shows that $\sval_v^{i+1}$ is fully determined by its boundary points with the same $B$-coordinates.
 Due to the shared boundaries, these two finite representations are equal to a common set $G$.
 Finally, every point in $G$ corresponds to a point in $G_w$ for some $w \in \succ(v)$. 
 By the induction hypothesis, $G_w$ has at most $3|V|^i$ points for every $w \in \succ(v)$,
 hence $G$ has at most $3|V|^{i+1}$ points.
\end{proof}

\newtheorem*{C1}{Corollary \ref{thm:bounded:determinacy}}
\begin{C1}
 Let $\tup{M,v,T,B,p}$ be an arbitrary problem instance.
 For every $i$, the point $(B,p)$ belongs to at least one of the sets $\rval_{v}^i$ and $\sval_{v}^i$.
\end{C1}

\begin{proof}
   Let $(B,p)$ by a point that does not belong $\sval_v^i$, and
   let $X$ be the intersection of the segment connecting $(B,p)$ with $(0,1)$ and the boundary of $\sval_{v}^i$.
   Note that such a point exists since $(0,1)$ belongs to $\sval_v^i$, and since $\sval_v^i$ is $\ord$-upward closed, by Lemma \ref{lem:fin-rep}.
   By Lemma \ref{lem:fin-rep}, every boundary point of $\sval_{v}^i$ belongs to $\rval_{v}^i$, especially $X$.
   Since $\rval_v^i$ is $\ord$-downward closed and $(B,p) \ord X$, we conclude $(B,p) \in \rval_v^i$.
\end{proof}

\subsection{Limiting Behavior of Value Iteration}



We prove that the limits of the value iteration sets form a fixpoint of the Bellman operator $\bellman$.

\begin{lemma}
    \label{lem:lim-fixpoint}
 For every vertex $v$, the following holds:
    \[
 \coptcrv_v = \bellman_v\left(\set{\coptcrv_w\mid w\in\succ(v)}\right).
    \]
\end{lemma}
\begin{proof}
 Let us first show the inclusion ``$\subseteq$''.
 If $v \in \VC$, then, by definition of $\coptcrv_v$, every $(B, p)$ in $\coptcrv_v$ also belongs to $\coval^i_v$ for every $i$.
 By inductive definition of $\coval^i_v$, for every $i > 0$, there exists $w_- \in \succ(v)$, $(B^i_-, p) \in \coval^{i-1}_{w_-}$, and $(B^i_+, p) \in \coval^{i-1}_{w_+}$ for every $w_+ \in \succ(v)$
 such that $B = (B^i_- + B^i_+)/2$. Since the sets $\coval^i$ monotonically decrease (in inclusion) with $i$ and $[0,1]$ is compact, we can without loss of generality assume that $B^i_-$ and $B^i_+$ converge to some $B_-$ and $B_+$,
 and that $(B_- + B_+)/2 = B$.
 Moreover, the limits $\lim_{j\to\infty} (B^j_-, p) = (B_-, p)$ and $\lim_{j\to\infty} (B^j_+, p) = (B_+, p)$ belong to $\coval^i_{w_-}$, respectively, to $\coval^i_{w_+}$ for every $w_+ \in \succ(v)$
 since $\coval^i_w$ is closed for every $w$.
 Hence $(B_-, p) \in \coptcrv_{w_-}$ and $(B_+, p) \in \coptcrv_{w_+}$ for every $w_+ \in \succ(v)$,
 and we conclude that $(B, p) = \left((B_- + B_+)/2, p\right) \in \bellman_v\left(\set{\coptcrv_w\mid w\in\succ(v)}\right)$.
 The proof for $v \in \VR$ is similar.

 Let us now show the inclusion ``$\supseteq$'', and again first consider the case $v \in \VC$.
 Whenever $(B, p)$ belongs to $\bellman_v\left(\set{\coptcrv_w\mid w\in\succ(v)}\right)$,
 there exist $B_-$ and $B_+$ such that $B = (B_- + B_+)/2$,
    $(B_-, p) \in \coptcrv_{w_-}$ for some $w_- \in \succ(v)$, and $(B_+, p) \in \coptcrv_{w_+}$ for every $w_+ \in \succ(v)$.
 By definition of $\coptcrv$, it holds that $(B_-, p) \in \coval^i_{w_-}$ and $(B_+, p) \in \coval^i_{w_+}$ for every $w_+ \in \succ(v)$ and for every $i \in \N$.
 By inductive definition of $\coval^i_v$, we derive $(B, p) \in \coval^i_v$ for every $i \in \N$. Therefore, $(B, p)$ also belongs to $\coptcrv_v$.
 The proof for $v \in \VR$ is again similar.
\end{proof}

\newtheorem*{C2}{Theorem \ref{cor:determinacy}}
\begin{C2}[Determinacy]
 For every MDP $M$, target vertices $T$, and vertex $v$ in $M$,
 it holds that $\optcrv_v \cap \coptcrv_v=\thres_v$.
 Moreover, the threshold is completely separating.
 \end{C2}
 
 \begin{proof}
 By Lem.~\ref{lem:fin-rep} all sets $\sval_v^i$ are $\ord$-downward closed, hence their limit $\coptcrv_v$ is $\ord$-downward closed as well.
 Similarly, $\optcrv_v$ is $\ord$-upward closed.
 Therefore, by Thm.~\ref{thm:limiting values}, every point $(B, p)$ in $\optcrv_v \cap \coptcrv_v$ belongs to $\thres_v$.
 It further follows from Cor.~\ref{thm:bounded:determinacy} and the definitions of $\optcrv_v$ and $\coptcrv_v$ that every point of $[0,1]^2$ is either in $\optcrv_v$ or in $\coptcrv_v$.
 Since both sets $\optcrv_v$ and $\coptcrv_v$ are closed, every point on the boundary of one lies in the other.
 Moreover, by Thm.~\ref{thm:limiting values}, the interiors of $\optcrv_v$ and $\coptcrv_v$ are disjoint,
 therefore $\optcrv_v \cap \coptcrv_v$ is necessarily the whole threshold.
 
 Finally, let $(B, p) \not\in \thres_v$. Suppose $(B,p) \in \optcrv_v$. Since $\optcrv_v$ and $\coptcrv_v$ are closed and share boundary points,
 the segment connecting $(B,p)$ to $(0,1)$ contains a threshold point. Similar argument works for $(B,p) \in \coptcrv_v$ and $(1,0)$. 
 Thus the threshold is completely separating.
 \end{proof}

\subsection{Speed of Convergence of Value Iteration}

In what follows, we will use the word history, usually denoted as $h$, and the word path interchangeably.
Given $h = v^0 \ldots v^\ell$, we will use $\last(h)$ to denote the last vertex $v^\ell$, $|h| = \ell$ to denote its length,
and $h^i$ to denote the $i$-th vertex $v^i$ of $h$.

We also naturally extend the meaning of the notation $\winpolZ(B, p, h)$ to an arbitrary history $h$ as the set
of policies $\polZ$ such that for every policy $\polO$ of the safety player, the probability of reaching $T$ under $\polZ$ and $\polO$ from $h$ is at least $p$,
given the budget in $h$ is $B$.

Let $\polZ, \polO$ be policies of the reachability and the safety player, respectively.
Given a history $h$ and the current budget $B_h$, we call a path $H$
a \emph{reference path} from $h$ if
\begin{itemize}
    \item $h$ is a prefix of $H$,
    \item $H$ terminates in a node in $T$,
    \item the suffix of $H$ starting at $h$ is consistent with the policy $\polZ$, i.e.,
 for every $i\geq |h|$ such that $H^i \in \VC$, the token moves according to $\polZ$.
\end{itemize}

\begin{lemma}[Short reference path existence]
    \label{lem:short-ref-path}
 Let $\tup{M,v,T,B,p}$ be a problem instance and suppose $(B, p)$ is an interior point of $\optcrv_{v}$.
 Then there exists a policy $\polZ \in \winpolZ(B, p, v)$ such that for every history $h$, current budget $B_h$, and $p_h > 0$,
 the fact $\polZ \in \winpolZ(B_h, p_h, h)$ implies that
 there exists a reference path (with respect to $\polZ$) from $h$.
\end{lemma}
\begin{proof}
 We will show that whenever $p>0$, then there exists a policy $\polZ \in \winpolZ(B, p, v)$ such that the reference path from the initial vertex $v$ exists.
 The general statement can be obtained by indictive construction of the policy;
 whenever the token deviates from the reference path, we consider the current vertex to beto be a new initial vertex and repeat the construction.

 Let $M^*$ be the infinite unfolding of $M$ and let us define $T^*, \VC^*,$ and $\VR^*$
 as the set of paths in $M^*$ that end in $T, \VC,$ and $\VR$, respectively.
 We call a finite subtree $F$ of $M^*$ \emph{successor-closed} if for every vertex $h \in F$ either
    $h$ is a leaf or $F$ contains all successors of $h$ in $M^*$.
 We call a tuple $(\polZ, F)$, where $F$ is a successor-closed finite subtree of $M^*$,
 a \emph{tree witness} for $(B, p)$ whenever for every policy $\polO$,
    $$\pr{v, B}{\polZ,\polO}(\reach^F(T^*))\geq p.$$
 Clearly, whenever a tree witness exists for $(B, p)$, the corresponding policy $\polZ$
 can be extended (by defining arbitrary actions in the vertices not in $F$) to a policy $\polZ'$ that belongs to $\winpolZ(B, p, v)$.

 If $(B, p)$ is a point in the interior of $\optcrv$, then it belongs to a $\rval^i_v$ for some $i$, hence
 already the first $i$ levels of $M^*$ together with the corresponding policy from Thm. \ref{thm:bounded value iteration} form a tree witness for $(B, p)$.
 Let moreover $(\polZ, F)$ be a tree witness for $(B, p)$ such that $F$ has the minimum number of vertices among all witnesses.

 Consider a path $h^0 = v, h^1, \ldots, h^\ell$ in $F$ (here $h^i$ is a vertex in $F$, i.e., a hisotry in $M$)
 such that $h^\ell$ is a leaf in $F$.
 Assume the reachability player wins all the biddings on the path or in other words,
 the vertex $h^{i+1}$ is the vertex chosen by $\polZ$ given $h^i$ and the current budget $B^i$, whenever $h^i \in \VC^*$.
 Note that the budgets $B^0, \ldots, B^\ell$ form a decreasing sequence
 since the reachability player always wins the bidding, thus loses the bid amount in the control vertices,
 and the budget does not change in the random vertices.
 We will further use $p(h)$ to denote the infimum of the reachability probability of $T^*$ from $h$
 over all policies $\polO$ of the safety player.

 While $h^{i+1}$ is uniquely determined by $\polZ$ provided $h^i \in \VC^*$,
 there is a freedom in the choice of the stochastic results when $h^i \in \VR^*$.
 We fix a unique path $h^0, \ldots, h^\ell$ by further requiring that whenever $h^i \in \VR^*$,
 then $h^{i+1}$ is any fixed successor of $h^i$ that maximizes the probability $p(h)$ among all successors $h$ of $h^i$.
    
 For $h^i \in CV^*$, the safety player can always lose the bidding willingly, thus getting the game to $h^{i+1}$, which shows that in such a case
    $p(h^i)$ is at most $p(h^{i+1})$.
 For $h^i \in \VR^*$, we know the value $p(h^i)$ is a convex combination of $p(h)$ over all successors $h$ of $h^i$.
 Since the value $p(h^{i+1})$ is the largest among all successors of $h^i$, it must be at least $p(h^{i})$.
 Therefore, the sequence $p(h^0), p(h^1), \ldots, p(h^\ell)$ is non-decreasing.
 In particular, since $p(h^\ell) \in \set{0,1}$ (depending on whether $h^\ell \in T^*$) and $p(h^0) \geq p > 0$, $p(h^\ell)$ is necessarily 1 and $h^\ell$ belongs to $T^*$.

 We will further show that $\ell \leq |V|$. Suppose for a contradiction that $\ell > |V|$.
 Then there are two indices $i < j$ such that $\last(h^i) = \last(h^j)$ by the pigeonhole principle.
 But then since $p(h^i) \leq p(h^j)$, and $B^i \geq B^j$, we could replace the whole subtree of $h^i$ by that of $h^j$ to obtain $F'$.
 By defining $\polZ'$ to be the same as $\polZ$ out of the subtree of $h^i$, and to be the same as $\polZ$ in the subtree of $h^j$ otherwise,
 we get a tree witness $(\polZ', F')$ for $(B, p)$ with fewer vertices than $F$, which is a contradiction.
 We conclude there exists the required reference path from $h = v$.
\end{proof}

\newtheorem*{L1}{Lemma \ref{lem:convergence speed}}
\begin{L1}[Convergence speed]
 Let $\tup{M,v,T,B,p}$ be a problem instance and
 suppose $(B,p)$ is in the interior of $\rval_{v}^*$ with its distance from $v$'s threshold being $\epsilon>0$.
 Then for every $n_\epsilon$ such that
       $$n_\epsilon \geq 4|\V|\log\left(2/\epsilon\right)\ER_{\min}^{-2|\V|},$$
 it holds that $(B,p)\in \rval_{v}^{n_\epsilon}$.
\end{L1}


\begin{proof}

 We will construct a policy $\polZ \in \winpolZ(B, p, v, \reach^{M, n_\epsilon}(T))$.
 Since the $L_\infty$ distance of $(B, p)$ from the boundary of $\optcrv_{v}$ is $\varepsilon > 0$,
 the point $(B', p') = (B-\frac{\varepsilon}{2}, p+\frac{\varepsilon}{2})$ belongs to the interior of $\optcrv_{v}$ as well.
 By Lemma \ref{lem:short-ref-path}, there exists a policy $\polZ' \in \winpolZ(B', p', v)$ and a reference path $H$ that reaches $T$ in at most $|V|$ steps with positive probability if the reachability player wins all the biddings according to $\polZ'$.
 At any given history $h$ of the game with the current budget $\overline{B}_h$, the policy $\polZ'$ witnesses a future reachability probability $\overline{p}_h$ of $T$, i.e., $\pr{h, \overline{B}_h}{\polZ', \polO}(\reach^M(T)) \geq \overline{p}_h$
 for every policy $\polO$ of the safety player. Moreover, the values $\overline{B}_h$ can be increased to $B'_h$ and values $\overline{p}_h$ decreased to $p'_h$ so that it still holds that
    $\pr{h, B'_h}{\polZ', \polO}(\reach^M(T)) \geq p'_h$ for every policy $\polO$ of the safety player, and the following hold:
    $B'_v = B, p'_v = p$,
 \begin{align*}
        2B'_h &= \max_{w \in \succ(\last(h))}B'_{hw} + \min_{w \in \succ(\last(h))}B'_{hw}\\
 p'_h &= p'_{hw},
 \end{align*}
 whenever $\last(h) \in \VC$, and
 \begin{align*}
 p'_h &= \sum_{w\in\succ(\last(h))} \ER(\last(h))(w) \cdot p'_{hw}\\
 B'_h &= B'_{hw},
 \end{align*}
 whenever $\last(h) \in \VR$.

 We will construct a policy $\polZ$ based on the values $(B'_h, p'_h)$.
 In a history $h$, policy $\polZ$ will maintain the current budget $B_h$, initially $B$, and the current required probability $p_h$, initially $p$.
    
 Let us now define the multiplicative slack $\theta_h = (B_h - B_h')(p_h' - p_h)$
 between the point $(B_h, p_h)$ and the point $(B'_h, p'_h)$.
 The idea of the proof is to follow the policy $\polZ'$
 and either reach $T$ according to the reference path $H$ or deviate from $H$ and increase the slack $\theta$.
 Upon deviation, we will choose a new reference path that starts at the current vertex and repeat the process.
 We will show that in exponentially many steps, the token either reaches $T$ according to the reference path
 or the slack increases so much that $p_h$ becomes zero.

 Let $h$ be a history, $H$ be the current reference path so that $h$ is a prefix of $H$,
 and let $h_0$ be the history in which $H$ was chosen.
 First, consider the case when $\last(h) \in \VC$.
 Then $\polZ$ bids $$\frac{B_{hw_+}' - B_{hw_-}'}{2} + \frac{s}{p'_h - p_h}, \text{ where } s = \theta_{h_0}\ER_{\min}^{2|H|-2|h|+1},$$ to
 move the token to the respective $w_-$ that lies on the reference path $H$.
 This way, the token either stays on the reference path and the surplus $\theta$ decreases by $s$,
 or the token deviates from the reference path and the surplus increases by at least $s$.

 Now consider the case when $\last(h) \in \VR$. We will use $\ER_h$ to denote the probability of transitioning from $\last(h)$ to the next vertex on $H$ whenever $\last(h) \in \VR$.
 Let $w$ be the sampled successor of $\last(h)$. If $w$ lies on the reference path $H$, we increase the required probability
 so that the slack decreases by $t = \frac{1-\ER_h}{\ER_h}\theta_{h_0}\ER_{\min}^{2|H|-2|h|+1}$:
 \begin{align*}
 p_{hw} &= p_h - p'_h + p_{hw}' + \frac{t}{B_h - B'_h} \\
 B_{hw} &= B_h.
 \end{align*}
 Otherwise, if $w$ is not on the reference path, we increase the slack by the amount $s$:
 \begin{align*}
 p_{hw} &= p_h - p'_h + p_{hw}' - \frac{s}{B_h - B'_h} \\
 B_{hw} &= B_h.
 \end{align*}
 Note that since the ratio of $t$ and $s$ is reciprocal to the ratio of probabilities of staying on/leaving the reference path,
 it holds that $\sum_{w\in\succ(\last(h))} \ER(\last(h))(w) \cdot p_{hw} = p_h$.
 Hence if $\polZ$ can guarantee reaching $T$ with probability at least $p_{hw}$ from $hw$ regardless of $w$,
 then it guarantees reaching $T$ with probability at least $p_h$ from $h$.

 If $p_{hw}$ should be less than zero, we can stop the game since the zero probability of reaching $T$ can be witnessed in zero steps.
    
 Now suppose at some history $h$ the token deviates to a vertex $w$ that is not on the reference path $H$.
 Then the new slack is at least
 \begin{align*}
 \theta_{hw} &\geq \theta_{h_0}\left(1 - \sum_{i=|h_0|}^{|h|}\frac{1-\ER_{\min}}{\ER_{\min}}\ER_{\min}^{2|H|-2i+1} + \ER_{\min}^{2|H|-2|h|-1}\right)\\
        &\geq \theta_{h_0}\left(1 - \sum_{i=|h_0|}^{|h|}\ER_{\min}^{2|H|-2i} + \ER_{\min}^{2|H|-2|h|-1}\right)\\
        &\geq \theta_{h_0} \left(1 + \ER_{\min}^{2|\V|}\right).
 \end{align*}
 On deviation from the reference path $H$, $\polZ$ repeats the process with the new
 history of last deviation $h'_0 = hw$ and a new reference path $H'$.
 Since the slack increases by at least the factor of $1 + \ER_{\min}^{2|\V|}$ on every deviation,
 after $k$ deviations, the slack is at least $\theta \cdot \left(1 + \ER_{\min}^{2|\V|}\right)^k$
 hence after at most
    $$k \geq \frac{\log(4/\varepsilon^2)}{\log(1 + \ER_{\min}^{2|\V|})}$$
 repetitions, the slack is at least one.
 For $\alpha \leq \frac{1}{4},$ we have $\log(1+\alpha) \geq \alpha/2$ so we can put the following upper bound on the above expression:
    $$k \geq 4\log(2/\varepsilon) \ER_{\min}^{-2|\V|}.$$
    
 We derived that under the policy $\polZ$,
 on every trajectory of length $k|V|$, the slack increases to at least one, which is not possible.
 Hence in less than $k|V|$ steps, the token either reaches $T$ or $p_h$ decreases to zero.
 We conclude that $$\polZ \in \winpolZ(B, p, v, \reach^{M, k \cdot |V|}(T)).$$
\end{proof}

\subsection{Abstraction-Based Value Iteration}

\newtheorem*{L2}{Lemma \ref{lem:abstraction precision}}

\begin{L2}
 For every $v\in\V$ and every $i$, $\abs{\sval}_v^i \supseteq \sval_v^i$, and, moreover, $\dH(\abs{\sval}_v^i, \sval_v^i)\leq \alpha (i+1)$.
 The abstract value $\abs{\sval}_v^i$ can be computed in $2|\V|^2i/\alpha$ time and $2|\V|/\alpha$ space.
\end{L2}
   
   
 \begin{proof}
 The fact that the operator $\bellman_v$ is monotonic directly follows from its definition.
 Since $\grid(X)$ is always a superset of $X$, by induction it follows that $\abs{\sval}_v^i \supseteq \sval_v^i$ for every $i$.
   
 We will now show that the approximation error grows linearly with $i$, namely that it is at most $\alpha (i+1)$.
 First note that due to the triangular inequality, the application of $\grid$ to a set $X$ can increase its Hausdorff distance to another set $Y$ by at most $\alpha$.
 Hence the base case when $i=0$ holds.
   
 Now suppose $i>0$. We only consider the case $v\in\VC$ since the proof for $v\in\VR$ is analogous.
 By the above argument, the operator $\grid$ will not increase the distance by more than $\alpha$, hence it is enough 
 to show that $\setdist(\bellman_v\left(\set{\abs{\sval}_w^{i-1}\mid w\in\succ(v)}\right), \sval_v^i) \leq \alpha i$.
 Namely, for every point $(B, p) \in \bellman_v\left(\set{\abs{\sval}_w^{i-1}\mid w\in\succ(v)}\right)$, we have to prove the existence of a point $(B', p') \in \sval_v^i$ such that $\lVert (B, p) - (B', p')\rVert_{\infty} \leq \alpha i$
 (the other direction is trivial since the two sets are in an inclusion).
 By definition of $\bellman_v$, the point $(B, p)$ belongs to $\bellman_v\left(\set{\abs{\sval}_w^{i-1}\mid w\in\succ(v)}\right)$ only if there exist $(B_-, p) \in \abs{\sval}^{i-1}_{w_-}$ for some $w_- \in \succ(v)$ and $(B_+, p) \in \abs{\sval}^{i-1}_{w_+}$ for every $w_+ \in \succ(v)$
 such that $B = \frac{B_+ + B_-}{2}$.
 Further, by the induction hypothesis, it holds $d_{\infty}\left((B_-, p), \sval_w^{i-1}\right) \leq \alpha i$ and $d_{\infty}\left((B_+, p), \sval_w^{i-1}\right) \leq \alpha i$ for every $w\in\succ(v)$.
 These inequalities are witnessed by points $(B_{w_-}, p_{w_-}) \in \sval_{w_-}^{i-1}$ and $(B_{w_+}, p_{w_+}) \in \sval_{w_+}^{i-1}$ for every $w_+$, respectively.
 Let $p_{\max}$ be maximal value among $p_{w_-}$ and all $p_{w_+}$, and $B_{\min}$ be the minimal value among all $B_{w_+}$.
 Since the sets $\sval_{w_-}^{i-1}$ and $\sval_{w_+}^{i-1}$ are $\ord$-upward closed,
 it holds that $(B_{w_-}, p_{\max}) \in \sval_{w_-}^{i-1}$ and $(B_{\min}, p_{\max}) \in \sval_{w_+}^{i-1}$ for every $w_+$.
 Since all differences $\lvert p - p_{\max}\rvert$, $\lvert B_- - B_{w_-}\rvert$, and $\lvert B_+ - B_{\min}\rvert$ are less or equal to $\alpha i$, it follows that the point $\left(\frac{B_{w_-} + B_{\min}}{2}, p_{\max}\right)$
 belongs to $\sval_v^i$ and is at most $\alpha i$ away from $(B, p)$.

 To represent the abstract value, it is enough to store $2/\alpha$ of its boundary grid points, hence all abstract values require $2|\V|/\alpha$ space.
 One iteration of the algorithm computes $|\V|$ abstract values, and for each of them sequentially combines up to $|\V|$ abstract values,
 which requires $2|\V|^2/\alpha$ time. The time per iteration is thus $2|\V|^2/\alpha$.
\end{proof}

\subsection{Bidding Games on Restricted MDP-s}

\newtheorem*{L4}{Lemma \ref{lemma: valueiterforDAG}}

\begin{L4}
	For acyclic MDPs, for every vertex \(v\), \(\rval^{|\V|}_v = \optcrv_v\) and $\sval^{|\V|}_v = \coptcrv_v$.
\end{L4}

\begin{proof}
	Let $d_v$ be the length of the longest simple path leading from a vertex $v$.
	We prove by induction on $d_v$ that $\val_v^{d_v} = \val_v^{d_v+k}$ for every $k \in \N$.
	If $d_v = 0$, then $v$ does not have any successors (except itself), hance $\val_v^i = \val_v^0$ in for every $i \in \N$.
	Now suppose $d_v = i+1$ and that the claim holds for every vertex $u$ such that $d_i \leq i$.
	Since $d_w < d_v$ for every successor $w$ of $v$, we have $\val_w^{d_w} = \val_w^{i} = \val_w^{i+k}$.
	We derive $\val_v^{i+1} = \val_v^{i+1+k}$ for every $k \in \N$.

	Since the sequences $\val_v^i$ stabilize after at most $|V|$ steps, the respective limits are equal to $\val_v^{|V|}$.
\end{proof}

We now further analyse Thm.~\ref{thm:bounded value iteration} to preciselly determine which of the points $(B,p)$ in sets $\rval_v^i$ and $\sval_v^i$ are actually achievable, i.e., when it is true that
$\winpolZ(B, p, v, \reach^{m, i}(T))$, respectively, $\winpolO(B, p, v, \reach^{m, i}(T))$ is nonempty.

\begin{lemma}
   \label{lem:dichotomy}
   Let $\tup{M,v,T,B,p}$ be a problem instance. The following hold:
   \begin{itemize}
      \item $\winpolZ(B, p, v, \reach^{M, i}(T)) \neq \emptyset$ if and only if there exists $B' < 1$ such that $B' \leq B$ and $(B', p) \in \rval_v^i$
      \item $\winpolO(B, p, v, \reach^{M, i}(T)) \neq \emptyset$ if and only if there exists $B' > B$ such that $(B', p) \in \sval_v^i$, or $B=1$ and $(B, p) \in \sval_v^i$.
   \end{itemize}
\end{lemma}
\begin{proof}
   According to Thm.~\ref{thm:bounded value iteration}, all the interior points of $\rval_v^i$ are winnable for the reachability player.
   On the other hand, the outerior points $(B,p)$ of $\rval_v^i$ are not winnable for the reachability player, since they are in the interior of $\sval_v^i$.
   Thus, for such points, the safety player is able to guarantee a reachability probability strictly less than $p$ with budget $B$.
   An analogous arguments hold for the safety player.
   
   Hence for the rest of the proof we suppose that $(B, p)$ lies on the threshold, i.e.,
   the intersection of $\rval_v^i$ and $\sval_v^i$.

   Let us first discuss whether $\winpolZ(B, p, v, \reach^{M, i}(T))$ is nonempty. 
   We prove the implication from left to right first.
   If $B<1$, then we can just set $B' = B$ and since $(B, p)$ is on threshold, the claim follows.
   If $B=1$ then note, that on a finite horizon of length $i$ already the budget $1-2^{-i}$ is enough for the reachability player to win all the bids,
   hence any higher budget does not bring any benefit in terms of increase of the reachability probability.
   Thus whenever $\winpolZ(B, p, v, \reach^{M, i}(T))$ is nonempty and $B=1$, it must be indeed the case that already $\winpolZ(B', p, v, \reach^{M, i}(T))$ is nonempty
   for $B' = 1-2^{-i}$, and we have already proven for such $B'$ that $(B', p) \in \rval_v^i$.

   For the opposite implication suppose there exists $B' < 1$ such that $ B' \leq B$ and $(B', p) \in \rval_v^i$. It is enough to show that the set
   $\winpolZ(B', p, v, \reach^{M, i}(T))$ is nonempty. Note that for $i=0$ the claim is true since any point $(B', p) \in \rval_v^0$ is winnable for the reachability
   player if $v \in T$, and any point $(B', 0)$ is winnable if $v \not \in T$.
   
   For the case $i>0$, we revisit the inductive construction
   from the proof of Thm.~\ref{thm:bounded value iteration}.
   The case $v \in \VR$ is the same as in the proof of Thm.~\ref{thm:bounded value iteration}. So suppose $v \in VC$,
   and let $B_+$ and $B_-$ be such that $(B_+ + B_-)/2 = B'$, $(B_+, p) \in \rval_{w_+}^{i-1}$ for every $w_+ \in \succ(v)$
   and $(B_-, p) \in \rval_{w_-}^{i-1}$ for some $w_- \in \succ(v)$.
   We can further without loss of generality assume $B_- \leq B$,
   hence by inductive hypothesis, there exists a policy in $\winpolZ(B_-, p, w_-, \reach^{m, i-1}(T))$.
   For every $w_+$, there exists a policy in $\winpolZ(B_+, p, w_+, \reach^{m, i-1}(T))$ whenever $B_+ < 1$.
   If the reachability player bids $|B_+ - B_-|/2$, then he either wins the bidding and decreases his budget to $B_-$
   or loses the bidding and increases his budget to $B_+$. However note, that since the reachability player wins
   the ties, the budget can only increase to $B_+$ if $B_+ < 1$. Therefore, regardless of the result, the reachability player
   can always continue according to some winning policy that exists by the inductive hypothesis.

   Now let us discuss when $\winpolO(B, p, v, \reach^{M, i}(T))$ is nonempty.
   We first show the direction from right to left. If there exists $B' > B$ such that $(B', p) \in \sval_v^i$, then the claim follows from 
   Thm.~\ref{thm:bounded value iteration}. So suppose $B=1$ and $(1, p) \in \sval_v^i$. Note that $B=1$ implies the initial safety player budget is zero.
   However since even the initial safety player budget $2^{-i}$ would not prevent the reachability player from winning all the bids in the following $i$ steps,
   the policy in $\winpolO(1-2^{-i}, p, v, \reach^{M, i}(T))$ proves there also exists a policy in $\winpolO(1, p, v, \reach^{M, i}(T))$.

   We show the other direction by a contraposition. Suppose $(B, p)$ lies on the threshold, $B < 1,$ and there is no $B' > B$ such that $(B', p) \in \sval_v^i$.
   Then necessarily, due to the staircase shape of $\sval_v^i$ and $\rval_v^i$, there is a point $(B, p')$  with $p' > p$ such that $(B, p') \in \rval_v^i$,
   hence $\winpolZ(B, p', v, \reach^{M, i}(T))$ is nonempty. This implies there is indeed no policy that would guarantee reachability probability at most $p$ with budget $B$.
\end{proof}

\newtheorem*{T4}{Theorem~\ref{thm:acyclic MDP}}
\begin{T4}
	For acyclic MDPs, \AlgExact is a sound and complete algorithm for Prob.~\ref{prob:exact}
	and terminates in at most $|\V|$ iterations, yielding the space and time complexity $\mathcal{O}\left(|\V|^{|\V|}\right)$.
\end{T4}
\begin{proof}
   The proof follows right from Lem.~\ref{lem:dichotomy} since the full knowledge of the value iteration sets is enough to decide the existence of the winning policy.
\end{proof}

\newtheorem*{T5}{Theorem~\ref{thm: treeshapedMDP}}
\begin{T5}
	For tree-shaped MDPs, Prob.~\ref{prob:exact} is in NP \(\cap\) co-NP.
\end{T5}
\begin{proof}
   Note that according to Lem.~\ref{lem:dichotomy}, the existence of a policy in $\winpolZ(B, p, v, \reach^{M, |\V|}(T))$
   or $\winpolO(B, p, v, \reach^{M, |\V|}(T))$ can be always witnessed by a presence of a point in $\rval_v^{|\V|}$ and $\sval_v^{|\V|}$, respectively.
   Moreover, the witness can be chosen so that it is a corner point of $\rval_v^{|\V|}$. Therefore, its binary representation requires
   at most $2|V|$ bits. Hence a nondeterministic algorithm can guess the witness point and solve the linear program (with nondeterministic choices) in Section \ref{sec:treeshapedMDP}.
   This shows that Prob.~\ref{prob:exact} is in NP.

   The idea of the proof that Prob.~\ref{prob:exact} is in co-NP is similar. In particular,
   to show that $\winpolZ(B, p, v, \reach^{M, |\V|}(T))$ is empty, it is enough to show that either $B = 1$ and there exists $(1, p') \in \sval_v^{|\V|}$ such that $p' < p$,
   or $B\neq 1$ and there exists $(B',p') \in \sval_v^{|\V|}$ such that $B' > B$ and $p' < p$.
   To show that $\winpolO(B, p, v, \reach^{M, |\V|}(T))$ is empty, it is enough to show that either $B \neq 1$ and there exists $(B, p') \in \rval_v^{|\V|}$ such that $p' > p$,
   or $B=1$ and there exists $(B', p') \in \rval_v^{|V|}$ such that $B' < B$ and $p' > p$. 
\end{proof}

%
%
%



\end{document}